\documentclass[12pt]{amsart}
\usepackage[margin=1in]{geometry}
\usepackage{amsfonts}
\usepackage{amssymb}
\usepackage[dvips]{graphics}
\usepackage{epsfig}
\usepackage{color}
\usepackage{tikz}
\usepackage{mathrsfs}
\usepackage [all,arc,curve,color, arrow, matrix,frame]{xy}

\usepackage[colorlinks ,pdfstartview=FitB]{hyperref}
 \newtheorem{thm}{Theorem}[section]
 \newtheorem{cor}[thm]{Corollary}
 \newtheorem{lemma}[thm]{Lemma}
 \newtheorem{prop}[thm]{Proposition}
 \theoremstyle{definition}
 \newtheorem{defn}[thm]{Definition}
 \theoremstyle{remark}
 \newtheorem{remark}[thm]{Remark}

 \numberwithin{equation}{section}
\numberwithin{figure}{section}

\newcommand{\norm}[1]{\left\Vert#1\right\Vert}
\newcommand{\set}[1]{\left\{ #1\right\} }

\newcommand{\C}{\mathbb{C}}
\newcommand{\R}{\mathbb{R}}

\newcommand{\Z}{\mathbb{Z}}

\newcommand{\SSS}{{\mathbb S}}

\newcommand{\N}{\mathbb{N}}

\newcommand{\M}{\mathbb{M}}

\renewcommand{\Re}{\operatorname{Re}}

\DeclareMathOperator{\supp}{supp}
\DeclareMathOperator{\spann}{span}

\DeclareMathOperator{\dist}{dist}

\def\idty{{\mathchoice {\mathrm{1\mskip-4mu l}} {\mathrm{1\mskip-4mu l}} %
{\mathrm{1\mskip-4.5mu l}} {\mathrm{1\mskip-5mu l}}}}

\newcommand\beq{\begin{equation}}
\newcommand\eeq{\end{equation}}

\newcommand{\Tr}{\operatorname{{\rm Tr}}}
\numberwithin{equation}{section}

\usepackage{pgf,tikz,pgfplots}
\pgfplotsset{compat=1.15}
\usetikzlibrary{arrows}
\begin{document}
\definecolor{qqqqff}{rgb}{0,0,1}

\title{Entanglement Entropy Bounds in the Higher Spin XXZ Chain}

\author[C. Fischbacher]{Christoph Fischbacher$^1$}
\address{$^1$ Department of Mathematics\\
	University of California, Irvine\\
	Irvine, CA, 92697, USA}
\email{fischbac@uci.edu}

\author[O. Ogunkoya]{Oluwadara Ogunkoya$^2$}
\address{$^2$ Department of Mathematics\\
	University of Alabama at Birmingham\\
	Birmingham, AL, 35294, USA}
\email{ogunkoya@uab.edu}

\date{\today}
%
\begin{abstract}
We consider the Heisenberg XXZ spin-$J$ chain ($J\in\mathbb{N}/2$) with anisotropy parameter $\Delta$. Assuming that $\Delta>2J$, and introducing threshold energies $E_{K}:=K\left(1-\frac{2J}{\Delta}\right)$, we show that the bipartite entanglement entropy (EE) of states belonging to any spectral subspace with energy less than $E_{K+1}$ satisfy a logarithmically corrected area law with prefactor $(2\lfloor K/J\rfloor-2)$.

This generalizes previous results by Beaud and Warzel \cite{BW18} as well as Abdul-Rahman, Stolz and one of the authors \cite{ARFS}, who covered the spin-$1/2$ case.

\end{abstract}

\maketitle


\section{Introduction}
In the following, we are going to prove upper bounds for the bipartite entanglement entropy of the spin--$J$ XXZ Heisenberg chain, where $J\in\{1/2,1,3/2,2\dots\}$. The spin--$1/2$ Heisenberg XXZ chain under the presence of a random magnetic field has recently attracted significant interest in the rigorous study of many--body--localization (MBL), where MBL phenomena such as exponential clustering  
 of correlations, zero--velocity Lieb--Robinson bounds and area laws for the entanglement entropy (EE) for states from the lowest energy regime (``droplet" regime) have been shown recently in \cite{BW17,BW18,EKS, EKS2017b}, for a recent survey of these results, cf.\ also \cite{StolzProc}. 

Recall that for one-dimensional models such as the chain, we say that the EE satisfies an area law with respect to the bipartition of the chain into a ``right" and a ``left" subchain if it is uniformly bounded in the size of the subchain. While an area law is generally considered as an indicator of MBL, many delocalized systems seem to exhibit a logarithmic correction, which means that the EE scales like the logarithm of the subchain's length. See the results in \cite{BW18, FSch, LeschkeSobolevSpitzer14,LeschkeSobolevSpitzer17,Mulleretal,mller2020stability,PfirschSobolev18,Wolf:2006ek}, where such logarithmic corrections to an area for the EE have been obtained. 

In what follows, we will adapt the ideas of Beaud and Warzel from \cite{BW18}, where a log--corrected area law for droplet states of the spin--$1/2$ Heisenberg XXZ chain for the generic case and a true area law under the presence of a disordered magnetic field was shown. With the help of technical refinements, this result was improved to show a log--corrected energy law also for higher--energy states in \cite{ARFS}.

While our results can certainly be viewed as a technical improvement of the spin-$1/2$ case, we nevertheless believe them to be of additional interest, since they demonstrate that given a discrete many--particle Schr\"odinger-type operator in one dimension, a logarithmically corrected area law for states in a given energy range follows if the following four criteria are met: (i) a suitable relative bound that controls the hopping operator in terms of the potential, (ii) the potential energetically favoring configurations with a lower number of ``building blocks", (iii) a sufficiently low dimension of the space of these building blocks and (iv) a generalized Pauli principle that limits the amount of particles that can occupy any site. 

We will proceed as follows:

In Section \ref{sec:2}, we will introduce the model and review previous results. We recall that the Spin-J XXZ model on the chain with anisotropy parameter $\Delta$ is unitarily equivalent to a direct sum of discrete $N$--particle Schr\"odinger operators of the form $-(2\Delta)^{-1}A_N+V_N$, where $A_N$ is a weighted  adjacency operator and $V_N$ is an interaction potential (Prop.\ \ref{prop:equivalent}). Compared to the previously studied spin-$1/2$ case, $V_N$ is much more complicated and we determine the set of all its minimizing configurations for sufficiently large $N$ (Prop.\ \ref{prop:minimizer}). 
Using that the kinetic term $A_N$ is controlled by the potential $V_N$ in the sense that $-4JV_N\leq A_N\leq 4JV_N$ (Prop. \ref{prop:relativebound}), we use a suitable Combes--Thomas estimate (Thm.\ \ref{prop:CT}) previously shown in \cite{ARFS}, which allows as to obtain decay estimates on spectral projections (Thm.\ \ref{thm:specestimate}). This together with some previous results on estimating the entanglement entropy (Lemma \ref{lemma:traceestimate}) reduces the problem to showing bounds for the sum of all possible $N$-particle configurations in the subsystem weighted by their distance to the nearest configuration with sufficiently low potential (Cor.\ \ref{cor:traceestimate2}).

Showing these bounds is the main goal of Section \ref{sec:3}: We firstly prove a suitable bound for single--cluster configurations in the spin-$1/2$ case (Lemma \ref{lemma:geom1d}), which we appropriately generalize to certain low--potential configurations for the spin-$J$ case (Lemma \ref{lemma:geom2J}). We then introduce the notion of ``building blocks" (Def.\ \ref{def:buildingblocks}) out of which any $N$--particle configuration can be composed. Exploiting that if an $N$-particle configuration has potential energy of $K$ or less, then it can be composed out of no more than $(\lfloor K/J\rfloor-1)$ building blocks, we obtain a suitable bound for configurations of arbitrary high potential energy (Lemma \ref{lemma:series}). Putting everything together, we prove our main result (Thm.\ \ref{thm:main}), where we show that states belonging to any spectral subspace with energy strictly less than $E_{K+1}=(K+1)
\left(1-\frac{2J}{\Delta}\right) $ satisfy a logarithmically corrected area law with prefactor $(2\lfloor K/J\rfloor-2)$.
\\\\
\noindent {\bf Acknowledgements:} C.F.~ is grateful to the Insitut Mittag-Leffler in Djursholm, Sweden, where some of this work was done as part of the program Spectral Methods in Mathematical Physics in Spring 2019. The authors would also like to thank G\"unter Stolz for many useful and insightful discussions and Shannon Starr for encouragement and support. It is also a pleasure to thank Abel Klein who suggested the much shorter proof of Theorem \ref{thm:specestimate} presented here.

 \section{Previous Results} \label{sec:2}
For any fixed $J\in\N/2$, we consider the chain of length $L\in\N\setminus\{1\}$, which is described by the Hamiltonian
\begin{equation} \label{eq:Hamiltonian}
H_L=\sum_{j=1}^{L-1}h_{j,j+1}+J(2J-S_1^3-S_L^3)
\end{equation}
acting on the Hilbert space $\mathcal{H}_L=\bigotimes_{j=1}^L\C^{2J+1}=\bigotimes_{j=1}^{L}\mathcal{H}_j$, with the interpretation of $\mathcal{H}_j=\C^{2J+1}$ as being the local Hilbert space describing a spin--$J$ particle located at site $j$. For later convenience, we also introduce the notation $\Lambda_L=\{1,2,\dots,L\}$. The two--site Hamiltonian $h_{j,j+1}$ is given by
\begin{equation}
h_{j,j+1}=J^2-S_{j}^3S_{j+1}^3-\frac{1}{\Delta}(S_j^1S_{j+1}^1+S_{j}^2S_{j+1}^2)=J^2-S_{j}^3S_{j+1}^3-\frac{1}{2\Delta}(S_j^+S_{j+1}^-+S_{j}^-S_{j+1}^+)
\end{equation}
and the additional term $J(2J-S_1^3-S_L^3)$ describes a boundary field; usually referred to as ``droplet boundary condition".
Here $S^1, S^2$ and $S^3$ are the spin--$J$ matrices and $S^\pm=(S^1\pm iS^2)/2$ are the spin raising and lowering operators, respectively (for the full definition of these matrices see e.g.\ \cite{F} or \cite{MNSS}). Labeling an operator $A\in\C^{(2J+1)\times(2J+1)}$ by a site $j$, like e.g. $A_j$, means that it acts as $A$ on the $j$-th factor of the tensor product and as the identity on all the other factors.
In what follows, we assume that the anisotropy parameter $\Delta$ satisfies $\Delta>2J$. While we emphasize that $\Delta>2J$ is not necessary for the operators $h_{j,j+1}$ to be non-negative (see, e.g.\ \cite[p.16]{Starr}), it is certainly sufficient (see Proposition \ref{prop:relativebound} below).

 \subsection{Equivalence to a direct sum of Schr\"odinger-type operators}
 In the following, we are going to discuss the equivalence of the spin-$J$ XXZ Hamiltonian to a direct sum of discrete many-particle Sch\"odinger-type operators. To this end, we will firstly review previous work \cite{F}, where we described configurations of particles via functions ${\bf m}:\Lambda_L\rightarrow\{0,1,\dots,2J\}$, where for each $i\in\Lambda_L$, the value ${\bf m}(i)$ represents number of particles that are located at site $i$ when in configuration $\bf{m}$. 
 
Equivalently, a configuration of $N$ particles distributed over $\Lambda_L$ can be described using ordered multisets $X$ with elements in $\Lambda_L$. Here, the value of the $i$-th element $x_i\in X$ represents the location of the $i$-th particle of a configuration of $N$ particles.

For our purposes, it will be advantageous to use both points of view as they have their respective advantages: the first description involving occupation numbers will be more useful when it comes to analyzing the interaction potential of the Schr\"odinger-type operators. On the other hand, once it has been shown how the entanglement entropy of a state can be estimated using the Combes-Thomas bound for spectral projections, the problem boils down to a combinatorial problem estimating exponentially weighted sums over a large set of many-particle configurations, for which the multiset point of view will be more convenient.

\subsubsection{Previous approach using occupation numbers}
In \cite[Prop.\ 2.1]{F}, it was shown that the Hamiltonian $H_L$ is unitarily equivalent to a direct sum of many--body Schr\"odinger operators. One firstly observes that $H_L$ preserves the total magnetization/particle number: to this end, we define the \emph{local particle number operator} $\mathcal{N}^{loc}:=(J-S^3)$ acting on $\C^{2J+1}$ which has spectrum $\sigma(\mathcal{N}^{loc})=\{0,1,\dots,2J-1,2J\}$. For any site $j\in\{1,\dots,L\}$, we interpret the eigenvalues of $\mathcal{N}^{loc}_j$ as the number of particles located at site $j$. We then define the \emph{total particle number operator} $\mathcal{N}_L$ as
\begin{equation}
\mathcal{N}_L:=\sum_{j=1}^L\mathcal{N}_j^{loc}\:,
\end{equation}
where the eigenvalues  $\sigma(\mathcal{N}_L)=\{0,1,2,\dots,2JL-1,2JL\}$ of $\mathcal{N}_L$ are consequently interpreted as the total number of particles. It can now be verified that $[H_L,\mathcal{N}_L]=0$, and thus we decompose $\mathcal{H}_L=\bigoplus_{N=0}^{2JL}\mathcal{H}_L^N$, where $\mathcal{H}_L^N$ denotes the eigenspace of $\mathcal{N}_L$ corresponding to the eigenvalue $N$ -- the space of all $N$-particle configurations in $\{1,2,\dots,L\}$ with the restriction that no site can be occupied by more than $2J$ particles. We also define $H_L^N:=H_L\upharpoonright_{\mathcal{H}_L^N}$. Now, let
\begin{equation}
{\bf{M}}^N_L:=\left\{{\bf{m}}:\{1,2,\dots,L\}\rightarrow \{0,1,\dots, 2J\}:\sum_{j=1}^L{\bf{m}}(j)=N \right\}\:,
\end{equation}
be the set of all functions from $\{1,2,\dots,L\}$ to $\{0,1,\dots,2J\}$ whose values add up to $N$. For convenience, we also define ${\bf M}_L:=\bigcup_{N=0}^{2JL}{\bf M}_L^N$ -- the set of \emph{all} functions from $\{1,2,\dots,L\}$ to $\{0,1,\dots,2J\}$. Let $\{e_k\}_{k=0}^{2J}$ denote a normalized eigenbasis of $\mathcal{N}^{loc}$, such that for any $k\in\{0,1,\dots,2J\}$, we have $\mathcal{N}^{loc}e_k=k\cdot e_k$. We then define for any ${\bf m}\in {\bf M}_L$
\begin{equation}
\psi_{\bf m}:=\bigotimes_{j=1}^L e_{{\bf m}(j)}\:.
\end{equation}
This means in particular that for any $j\in\{1,2,\dots,L\}$ we get $\mathcal{N}_j^{loc}\psi_{\bf m}={\bf m}(j)\psi_{\bf m}$. In other words, $\psi_{\bf m}$ describes a configuration of particles, where at each site $j\in\{1,2,\dots,L\}$, there are exactly ${\bf m}(j)$ particles. Since
\begin{equation}
\mathcal{N}_L\psi_{\bf m}=\left(\sum_{j=1}^L{\bf m}(j)\right)\psi_{\bf m}\:,
\end{equation}
it immediately follows that 
\begin{equation}
\mathcal{H}_L^N=\spann\left\{\psi_{\bf m}:{\bf m}\in {\bf{M}}^N_L\right\}\:.
\end{equation}
Now, consider the Hilbert space
$ \ell^2({\bf{M}}^N_L)=\left\{f:{\bf{M}}^N_L\rightarrow \C \right\}
$
equipped with inner product $\langle f,g\rangle=\sum_{{\bf m}\in {\bf{M}}^N_L}\overline{f({\bf m})}g({\bf m})$ and let $\{\phi_{\bf m}\}_{{\bf m}\in {\bf{M}}^N_L}$ denote the canonical basis of $\ell^2({\bf{M}}^N_L)$, i.e. 
\begin{equation} \label{eq:canbasis}
\phi_{\bf m}({\bf n})=\begin{cases} &1\quad \mbox{if} \quad {\bf m}={\bf n}\\ &0 \quad \mbox{else.}
\end{cases}
\end{equation}
The Hilbert spaces $\mathcal{H}_L^N$ and $\ell^2({\bf{M}}^N_L)$ are unitarily equivalent via 
\begin{equation} \label{eq:unitary}
U_L^N:\mathcal{H}_L^N\rightarrow\ell^2({\bf{M}}^N_L), \quad \psi_{\bf m}\mapsto \phi_{\bf m}\:.
\end{equation}
For any $f\in\ell^2({\bf{M}}^N_L)$, let us now define the adjacency operator $A_N$, given by
\begin{equation} \label{eq:adj}
(A_Nf)({\bf m})=\sum_{{\bf n}: {\bf n}\sim{\bf m}}w({\bf m},{\bf n})f({\bf n})\:,
\end{equation}
where for two configurations ${\bf m},{\bf n}\in{\bf{M}}^N_L$ to be adjacent (denoted by ${\bf m}\sim{\bf n}$) is defined as follows:
\begin{align}
{\bf m}\sim {\bf n}:\Leftrightarrow&\\ \exists j_0&\in \{1,2,\dots,L-1\}: {\bf m}(j_0)-{\bf n}(j_0)=\pm 1 \mbox{ and } {\bf m}(j_0+1)-{\bf n}(j_0+1)=\mp 1\notag\\
&\qquad\qquad\qquad\qquad\qquad\mbox{ and for any } j\in\{1,2,\dots,L\}\setminus\{j_0,j_0+1\}: {\bf m}(j)={\bf n}(j)\:.
\label{eq:adjacent}
\end{align}
This definition should be interpreted in the following way: two configurations ${\bf m,n}$ of $N$ particles distributed over $L$ sites (with the requirement that no site be occupied by more than $2J$ particles) are adjacent if one configuration can be obtained by moving a single particle from the other configuration to the right or left (cf.\ Figure \ref{fig:Z2disconnected}).
\begin{figure}[ht]
	\centering
	\resizebox{8cm}{6cm}
	{

\begin{tikzpicture}[line cap=round,line join=round,>=triangle 45,x=1cm,y=1cm]
\begin{axis}[
x=1cm,
axis x line=middle,
axis y line=none,
xmajorgrids=true,
xmin=0.5,
xmax=8.5,
ymin=-0.5,
ymax=5,
xtick={0,1,...,8},
]
\clip(0,0) rectangle (11.32,8.7);
\node at (.7, 3.5) {${\bf m}$};
\node at (.7, 1.5) {${\bf n}$};
\draw [line width=2pt] (1,1)-- (8,1);
\draw [line width=2pt] (1,3)-- (8,3);
\begin{scriptsize}
\draw [fill=qqqqff] (1,3) circle (2.5pt);
\draw [fill=qqqqff] (1,1) circle (2.5pt);
\draw [fill=qqqqff] (2,1) circle (2.5pt);
\draw [fill=qqqqff] (2,3) circle (2.5pt);
\draw [fill=qqqqff] (2,3.4) circle (2.5pt);
\draw [fill=qqqqff] (3,3) circle (2.5pt);
\draw [fill=qqqqff] (2,3.8) circle (2.5pt);
\draw [fill=qqqqff] (5,3) circle (2.5pt);
\draw [fill=qqqqff] (7,3) circle (2.5pt);
\draw [fill=qqqqff] (7,3.4) circle (2.5pt);
\draw [fill=qqqqff] (8,3) circle (2.5pt);
\draw [fill=qqqqff] (8,3.4) circle (2.5pt);
\draw [fill=qqqqff] (2,1.4) circle (2.5pt);
\draw [fill=qqqqff] (2,4.2) circle (2.5pt);
\draw [fill=qqqqff] (2,1.8) circle (2.5pt);
\draw [fill=qqqqff] (3,1) circle (2.5pt);
\draw [fill=qqqqff] (3,1.4) circle (2.5pt);
\draw [fill=qqqqff] (5,1) circle (2.5pt);
\draw [fill=qqqqff] (7,1) circle (2.5pt);
\draw [fill=qqqqff] (8,1) circle (2.5pt);
\draw [fill=qqqqff] (7,1.4) circle (2.5pt);
\draw [fill=qqqqff] (8,1.4) circle (2.5pt);
\node (a) at (2,4.15) {};
\node (b) at (3, 3.2) {};
\draw[->] (a)  to [out=75,in=100, looseness=1.5] (b);
\end{scriptsize}
\end{axis}
\end{tikzpicture}  
}
\caption{An example of two adjacent configurations ${\bf m, n}\in {\bf M}_8^{11}$ (here: $J= 2$). The values of the functions at a site are represented by blue circles corresponding to particles occupying the respective sites, e.g. ${\bf m}(2)=4$. Since ${\bf m}(2)-{\bf n}(2)=4-3=1$ and ${\bf m}(3)-{\bf n}(3)=1-2=-1$, while ${\bf m}(j)={\bf n}(j)$ for any other site $j$, we have ${\bf m}\sim{\bf n}$. We interpret configuration   ${\bf n}$ as obtained from configuration ${\bf m}$ by a particle hopping to the right from site $2$ to site $3$. (Represented by an arrow emanating from the hopping particle in configuration ${\bf m}$.)}
\label{fig:Z2disconnected}
\end{figure}

For ${\bf m}\sim{\bf n}$, the weight function $w({\bf m,n})=w({\bf n,m})$ in \eqref{eq:adj} is given by
\begin{equation} \label{eq:weight}
w({\bf m,n})=\prod_{j:{\bf m}(j)\neq{\bf n}(j)}\left(J({\bf m}(j)+{\bf n}(j)+1)-{\bf m}(j){\bf n}(j)\right)^{1/2}\:,
\end{equation}
however in what follows, the explicit expression in \eqref{eq:weight} will not play a particularly important role. Moreover, for any ${\bf m, n}\in {\bf M}_L^N$, we define their distance $d^N({\bf m, n})$ to be the length of the shortest path connecting ${\bf m}$ and ${\bf n}$, which we will refer to as graph distance.
\begin{remark}
For more details concerning the construction of the spin-$J$ Heisenberg XXZ model on more general underlying graph (in lieu of just the chain), we refer to \cite{F}, where appropriate $N$-particle graphs were introduced. Similar approaches have been used for the spin-$1/2$ case in \cite{FS18, Ouyang}. Consider also \cite{ART,Carballosa,Rivera} and the references therein, where similar constructions have been studied from a graph-theoretic point of view.
\end{remark}

Next, we define the full interaction potential $V_N$, which is a multiplication operator on $\ell^2({\bf M}_L^N)$, to be given by
\begin{equation} \label{eq:potential}
(V_Nf)({\bf m})=V({{\bf m}})f({\bf m})=\left(\sum_{j=1}^{L-1} v({\bf m}(j),{\bf m}(j+1))\right)f({\bf m})+J({\bf{m}}(1)+{\bf{m}}(L))f(\bf{m})\:,
\end{equation}
where the two--site potential $v$ is given by
\begin{equation}
v({\bf m}(j),{\bf m}(j+1))=J({\bf m}(j)+{\bf m}(j+1))-{\bf m}(j){\bf m}(j+1)
\end{equation}
and we refer to the extra term ``$J({\bf m}(1)+{\bf m}(L))$" as the ``boundary field".

In what follows, the following two facts will be particularly important: firstly that the operator $H_L^N$ is unitarily equivalent to the Schr\"odinger-type operator $-(2\Delta)^{-1}A_N+V_N$ (Prop.\ \ref{prop:equivalent}) and moreover that that kinetic term can be controlled in terms of the potential (Prop.\ \ref{prop:relativebound}):
\begin{prop}[{\cite[Prop.\ 2.1]{F}}] We have the following unitary equivalence:
\begin{equation}
U_L^NH_L^N(U_L^N)^*=-\frac{1}{2\Delta}A_N+V_N=:H_N\equiv H_N(L)\:.
\end{equation}
\label{prop:equivalent}
\end{prop}
\begin{proof}[Remark on the proof]
The only detail which is not discussed in \cite{F} is the unitary equivalence of the boundary field term ``$J({\bf{m}}(1)+{\bf{m}}(L))f(\bf{m})$" in \eqref{eq:potential} to the boundary field ``$J(2J-S_1^3-S_L^3)$" in \eqref{eq:Hamiltonian}, which can be verified by an easy calculation.
\end{proof}
\begin{remark}
For the special cases $N=0$ and $N=2JL$, note that $$\dim(\ell^2({\bf M}_L^0))=\dim(\ell^2({\bf M}_L^{2JL}))=1\:.$$ On these one-dimensional spaces, the operators $H_0$ and $H_{2JL}$ are just given by $H_0=0$ and $H_{2JL}=4J^2$.
\end{remark}
Let us now recall a useful relative bound of $A_N$ in terms of the potential $V_N$. It is because of this particular feature of the model that we do not have to worry about the explicit form of the weight function $w$ given in \eqref{eq:weight}.
\begin{prop}[{\cite[Lemma 2.9]{F}}]
The operators $A_N$ and $V_N$ satisfy the following relative bound:
\begin{equation}
-4JV_N\leq A_N\leq 4JV_N\:.
\end{equation}
\label{prop:relativebound}
\end{prop}
\begin{proof}[Remark on the proof] Strictly speaking, in {\cite[Lemma 2.9]{F}}, it was only shown that $A_N\leq 4JV_N$. However, the lower bound $-4JV_N\leq A_N$ follows from a completely analogous argument. 
	
\end{proof}

Let us now further analyze the interaction potential $V_N$ and determine all the configurations which minimize its value, the proof can be found in Appendix \ref{app:b1}
\begin{prop}\label{prop:minimizer}
Let $N\in\{4J,4J+1,\dots,2JL\}$ and define $V_{N,0}:=\min\{V({\bf m}):{\bf m}\in{\bf M}_L^N\}$. Then 
	\begin{equation}
	V_{N,0}:= 4J^2\:.
	\end{equation}
Moreover, -- up to overall translations -- the minimizers of $V_N$ are given by

\begin{equation} \label{eq:form1}
{\bf m}^N_j(x)= \begin{cases} j \quad&\mbox{if}\quad x=1\\
2J&\mbox{if}\quad x=2,\ldots,r\\
2J-j&\mbox{if}\quad x=r+1
\end{cases} 
\end{equation}
for  $N=2Jr,\;\;2\leq r\leq L-1$, $j=0,\ldots,2J-1$ and
\begin{equation} \label{eq:form2}
{\bf m}^N_j(x)= \begin{cases} j \quad&\mbox{if}\quad x=1\\
2J&\mbox{if}\quad x=2,\ldots,1+\lfloor\frac{N}{2J}\rfloor\\
N(\mbox{mod }2J)-j&\mbox{if}\quad x=2+\lfloor\frac{N}{2J}\rfloor
\end{cases} 
\end{equation}
if $N$ is not a multiple of $2J$, $j=0,\ldots, N(mod \;2J) -1$.
\end{prop}

\begin{remark} The following figure provides an example of a minimizer of the form \eqref{eq:form1}, i.e. when $N\in\{4J,6J,8J,\dots,2JL\}$:
\end{remark}
\clearpage
\begin{figure}[ht]
	\centering
	\resizebox{8cm}{6cm}
	{\begin{tikzpicture}[line cap=round,line join=round,>=triangle 45,x=1cm,y=1cm]
			\begin{axis}[
				x=1cm,y=1cm,
				axis x line=bottom,
				axis y line=none,
				xmajorgrids=true,
				xmin=0.2,
				xmax=8.5,
				ymin=-0.5,
				ymax=6.5,
				xtick={1,...,7},
				xticklabels={$k$,$k+1$, , \ldots, ,, $k+r$}
				]
				\clip(0,0) rectangle (11.32,8.7);
				\node at (1, 2.4) {$j$};
				\node at (2, 4.8) {$2J$};
				\node at (3, 4.8) {$2J$};
				\node at (4, 4.8) {$2J$};
				\node at (5, 4.8) {$2J$};
				\node at (6, 4.8) {$2J$};
				\node at (7, 3.8) {$2J-j$};
				\draw [line width=2pt] (0,1)-- (8,1);
				\begin{scriptsize}
					\draw [fill=qqqqff] (1,1) circle (2.5pt);
					\draw [fill=qqqqff] (1,1.4) circle (2.5pt);
					\draw [fill=qqqqff] (1,1.8) circle (2.5pt);
					\draw [fill=qqqqff] (2,1) circle (2.5pt);
					\draw [fill=qqqqff] (2,1.4) circle (2.5pt);
					\draw [fill=qqqqff] (2,1.8) circle (2.5pt);
					\draw [fill=qqqqff] (2,2.2) circle (2.5pt);
					\draw [fill=qqqqff] (2,2.6) circle (2.5pt);
					\draw [fill=qqqqff] (2,3) circle (2.5pt);
					\draw [fill=qqqqff] (2,3.4) circle (2.5pt);
					\draw [fill=qqqqff] (2,3.8) circle (2.5pt);
					\draw [fill=qqqqff] (2,4.2) circle (2.5pt);
					\draw [fill=qqqqff] (3,1) circle (2.5pt);
					\draw [fill=qqqqff] (3,1.4) circle (2.5pt);
					\draw [fill=qqqqff] (3,1.8) circle (2.5pt);
					\draw [fill=qqqqff] (3,2.2) circle (2.5pt);
					\draw [fill=qqqqff] (3,2.6) circle (2.5pt);
					\draw [fill=qqqqff] (3,3) circle (2.5pt);
					\draw [fill=qqqqff] (3,3.4) circle (2.5pt);
					\draw [fill=qqqqff] (3,3.8) circle (2.5pt);
					\draw [fill=qqqqff] (3,4.2) circle (2.5pt);
					\draw [fill=qqqqff] (4,1) circle (2.5pt);
					\draw [fill=qqqqff] (4,1.4) circle (2.5pt);
					\draw [fill=qqqqff] (4,1.8) circle (2.5pt);
					\draw [fill=qqqqff] (4,2.2) circle (2.5pt);
					\draw [fill=qqqqff] (4,2.6) circle (2.5pt);
					\draw [fill=qqqqff] (4,3) circle (2.5pt);
					\draw [fill=qqqqff] (4,3.4) circle (2.5pt);
					\draw [fill=qqqqff] (4,3.8) circle (2.5pt);
					\draw [fill=qqqqff] (4,4.2) circle (2.5pt);
					\draw [fill=qqqqff] (5,1) circle (2.5pt);
					\draw [fill=qqqqff] (5,1.4) circle (2.5pt);
					\draw [fill=qqqqff] (5,1.8) circle (2.5pt);
					\draw [fill=qqqqff] (5,2.2) circle (2.5pt);
					\draw [fill=qqqqff] (5,2.6) circle (2.5pt);
					\draw [fill=qqqqff] (5,3) circle (2.5pt);
					\draw [fill=qqqqff] (5,3.4) circle (2.5pt);
					\draw [fill=qqqqff] (5,3.8) circle (2.5pt);
					\draw [fill=qqqqff] (5,4.2) circle (2.5pt);
					\draw [fill=qqqqff] (6,1) circle (2.5pt);
					\draw [fill=qqqqff] (6,1.4) circle (2.5pt);
					\draw [fill=qqqqff] (6,1.8) circle (2.5pt);
					\draw [fill=qqqqff] (6,2.2) circle (2.5pt);
					\draw [fill=qqqqff] (6,2.6) circle (2.5pt);
					\draw [fill=qqqqff] (6,3) circle (2.5pt);
					\draw [fill=qqqqff] (6,3.4) circle (2.5pt);
					\draw [fill=qqqqff] (6,3.8) circle (2.5pt);
					\draw [fill=qqqqff] (6,4.2) circle (2.5pt);
					\draw [fill=qqqqff] (7,1) circle (2.5pt);
					\draw [fill=qqqqff] (7,1.4) circle (2.5pt);
					\draw [fill=qqqqff] (7,1.8) circle (2.5pt);
					\draw [fill=qqqqff] (7,2.2) circle (2.5pt);
					\draw [fill=qqqqff] (7,2.6) circle (2.5pt);
					\draw [fill=qqqqff] (7,3) circle (2.5pt);
				\end{scriptsize}
			\end{axis}
		\end{tikzpicture}
	}
	\caption{An example of a minimizer ${\bf m}_j^N(\cdot -k)$ when $N=2Jr$ as in \eqref{eq:form1}. Here, $J=9/2$ and $j=3$.}
	\label{fig:minimizer2}
\end{figure}
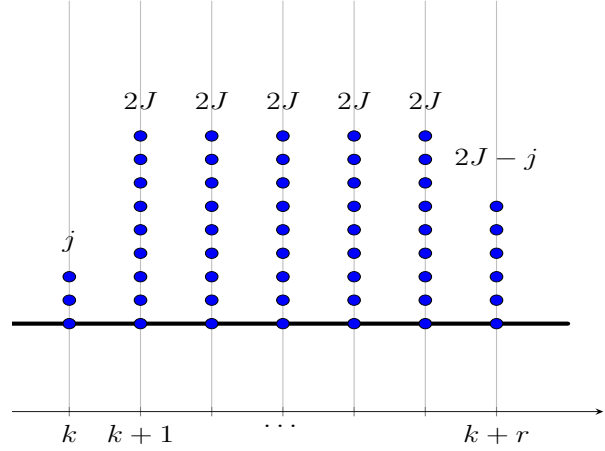
\begin{defn} For any $K\in\N$, let ${\bf M}_{L,K}^N=\{{\bf m}\in{\bf M}_L^N: V({\bf m})\leq K\}$, i.e.\ the set of occupation number functions for which the potential $V_N$ is bounded by $K$. \label{def:MLKN}
\end{defn}

\subsubsection{An equivalent description using multisets}

When it comes to tracking the positions of $N$ individual particles, it will be convenient to introduce ordered $N$-tuples that satisfy certain conditions. For brevity, in what follows, we will refer to these $N$-tuples as multisets. Thus, for any $L\in \N$ and any $N\in\{1,\dots,2JL\}$, let us define
\begin{equation} \label{eq:SSS}
\SSS_L^N=\left\{(x_1,x_2,\dots,x_N)\in\Lambda_L^N: x_1\leq x_2\leq\cdots\leq x_N\mbox{ and } \min_{k\in\{1,2,\dots, N-2J\}}(x_{k+2J}-x_k)\geq 1 \right\}\:,
\end{equation}
where for any $X=(x_1,x_2,\dots,x_N)\in\SSS_L^N$, the value of each individual $x_i$ represents the position of the $i$-th particle. The condition 
\begin{equation}
\min_{k\in\{1,2,\dots, N-2J\}}(x_{k+2J}-x_k)\geq 1
\end{equation} 
reflects the fact that no site can be occupied by more than $2J$ particles and is therefore automatically satisfied if $N\leq 2J$. For later convenience, we also introduce the convention $\SSS_L^0:=\{\emptyset\}$.

The correspondence between functions of occupation numbers ${\bf m}\in{\bf M}_L^N$ and $N$-particle configurations $X\in\SSS_L^N$ is of course straightforward. For a given $X=(x_1,x_2,\dots,x_N)\in\SSS_L^N$, the corresponding function ${\bf m}_X\in{\bf M}_L^N$ is defined as
\begin{equation} \label{eq:defm}
{\bf m}_X(j):=|\{k\in\{1,2,\dots,N\}:x_k=j\}|
\end{equation}
for any $j\in\Lambda_L$. It is not hard to see that the mapping $X\mapsto {\bf m}_X$ is a bijection from $\SSS_L^N$ to ${\bf M}_L^N$. We thus denote by $X_{\bf m}$ the image of the inverse of that mapping applied to an arbitrary ${\bf m}\in{\bf M}_L^N$.\footnote{As an example, consider the multiset corresponding to the occupation function ${\bf n}$ in Figure \ref{fig:Z2disconnected}, which would be given by
	\begin{equation}
	X_{\bf n}=(1,2,2,2,3,3,5,7,7,8,8)\:.
	\end{equation}} 
We will also use $\{\phi_X\}_{X\in\SSS_L^N}$ to denote the canonical basis of $\ell^2(\SSS_L^N)$, i.e. $\phi_X(Y)=1$ if $X=Y$ and $\phi_X(Y)=0$ if $X\neq Y$.
Then, the identification of the Hilbert spaces $\ell^2({\bf M}_L^N)$ with $\ell^2(\mathbb{S}_L^N)$ via $\phi_{\bf m}\mapsto \phi_{X_{\bf{m}}}$ is straightforward. So for any $X,Y\in\mathbb{S}_L^N$, we define $X\sim Y:\Leftrightarrow {\bf m}_X\sim {\bf m}_Y$ as well as $w(X,Y):=w({\bf m}_X,{\bf m}_Y)$  and $V(X):=V({\bf m}_X)$. By a slight abuse of notation, we will use the same symbols $H_N, A_N$ and $V_N$ to denote the unitarily equivalent operators on $\ell^2(\mathbb{S}_L^N)$, i.e.\
\begin{equation}
(H_Nf)(X)=-\frac{1}{2\Delta}(A_Nf)(X)+(V_Nf)(X)=-\frac{1}{2\Delta}\sum_{Y:X\sim Y}w(X,Y)f(Y)+V(X)f(X)
\end{equation}
for any $f\in\ell^2(\mathbb{S}_L^N)$.

For any $X,Y\in\SSS_L^N$, let $d^N(X,Y):=d^N({\bf m}_X,{\bf m}_Y)$. The merit of the multiset-point-of-view will now be made more apparent by the following lemma; its proof can be found in Appendix \ref{app:a}.
\begin{lemma}
	For any two configurations $X,Y\in\SSS_L^N$, where $X=(x_1,x_2,\dots,x_N)$ and $Y=(y_1,y_2,\dots,y_N)$, we have
	\begin{equation} \label{eq:distformula}
	d^N(X,Y)=\sum_{i=1}^N|x_i-y_i|\:.
	\end{equation}
\end{lemma}

 \subsection{Combes--Thomas estimate and bounds on spectral projections}

One of the main ingredients of the proof of the bound for the entanglement entropy will be a bound on spectral projections. To be more specific, let $\mathcal{A}\subset \mathbb{S}_L$ be a set of configurations. For any such $\mathcal{A}$, we define $P_{\mathcal{A}}$ to be the orthogonal projection onto the subspace of functions that are supported on $\mathcal{A}$, i.e.
\begin{equation}
(P_{\mathcal{A}}f)(X)=\begin{cases} f(X)\quad&\mbox{if}\quad X\in\mathcal{A}\\ 0 \quad&\mbox{if}\quad X\in\mathbb{S}_L\setminus\mathcal{A}\end{cases}\:.
\end{equation}
For later convenience, we also define $\overline{P}_{\mathcal{A}}:=\idty-P_{\mathcal{A}}$.
\begin{thm} \label{prop:CT}

For any $N\in\{1,2,\dots,2JL\}$, let $Y_N$ be an arbitrary multiplication operator on $\ell^2(\SSS_L^N)$ and $z\notin\sigma(H_N+Y_N)$ such that there exists $\kappa_z>0$ for which
	\begin{equation} \label{initbound}
	\left\|V_N^{1/2}(H_N+Y_N-z)^{-1}V_N^{1/2}\right\|\leq \frac{1}{\kappa_z}<\infty\:.
	\end{equation}
	Then for all subsets $\mathcal{A,B}\subseteq \SSS_L^N$, we have
	\begin{equation}
	\left\|P_\mathcal{A}\left(H_N+Y_N-z\right)^{-1}P_\mathcal{B}\right\| \leq \frac{1}{V_{N,0}} \left\|P_\mathcal{A}V_N^{1/2}\left(H_N+Y_N-z\right)^{-1}V_N^{1/2}P_\mathcal{B}\right\|\leq \frac{2}{V_{N,0}\kappa_z}\,e^{-\eta_z d^N(\mathcal{A,B})}\:,
	\end{equation}
	where
	\begin{equation} \label{etaz}
	\eta_z= \log\left(1+\frac{\Delta\kappa_z}{4J}\right).
	\end{equation}
\end{thm}
\begin{proof} An abstract result of this form for Schr\"odinger-type operators whose kinetic term is controlled by the potential such as in Proposition \ref{prop:relativebound} was shown in \cite[Prop.\ 3.1]{ARFS}. The theorem now follows from plugging in the constants particular to this model.
\end{proof}
Now, let $K\in\N$ and for any $\delta\in(0,1)$, define $E_{K,\delta}:=\left(1-\frac{2J}{\Delta}\right)(K+1-\delta)$. Moreover, for any non-negative multiplication operator $W_N$ on $\ell^2(\SSS_L^N)$, let $Q_{K,\delta}^N\equiv Q^N_{K,\delta}(L,W_N):=\idty_{[0,E_{K,\delta}]}(H_N(L)+W_N)$ be the spectral projection of $H_N+W_N$ associated to the energy interval $[0,E_{K,\delta}]$. Lastly, for any $K\in\N$, we introduce the set $\mathbb{S}_{L,K}^N:=\{X_{\bf m}: {\bf m}\in {\bf M}_{L,K}^N\}$, where the sets ${\bf M}_{L,K}^N$ were introduced in Definition \ref{def:MLKN}. 
\begin{thm}  Let $\mathcal{A}\subset \SSS_L^N$ be a set of configurations. We then get the following estimate:
	\begin{align} \label{eq:CTCT}
	\norm{ P_\mathcal{A}Q_{K,\delta}^N}=\norm{ Q_{K,\delta}^NP_\mathcal{A}}\le  C_{N,K} e^{-\mu_{K} d^N(\mathcal{A},\SSS_{L,K}^N)},
	\end{align}
	where
	\begin{align}\label{eq:fullCT2'2}
	C_{N,K} &=C_{N,K}(\Delta,\delta,J)=\max\left\{1,\frac{8K(K+1)}{V_{N,0}\delta^2}\right\}\\\quad\mbox{and}\quad \mu_{K} &= \mu_{K}(\Delta,\delta,J)=\log\left(1+\frac{\delta (\Delta-2J)}{16J(K+1)}\right)  .
	\end{align}
	\label{thm:specestimate}
\end{thm}

\begin{proof}  We follow ideas from \cite[Proof of Lemma 8.2]{EKS} and \cite{Abel}. If $\mathcal{A}\cap\SSS_{L,K}^N\neq\emptyset$, which implies $d^N(\mathcal{A},\SSS_{L,K}^N)=0$, then \eqref{eq:CTCT} immediately follows from $C_{N,K}\geq 1$. Hence, assume ${\mathcal{A}}\subset \SSS_L^N\setminus\SSS_{L,K}^N$ from now on.  Then, let us choose $Y_N=W_N+\gamma P_{\SSS_{L,K}^N}$ with $\gamma=\left(1-\frac{2J}{\Delta}\right)K$ and show that the operator
$H_N+W_N+\gamma P_{\SSS_{L,K}^N}$ satisfies the assumptions of Theorem \ref{prop:CT}. Observe that for any $E\in[0,E_{K,\delta/2}]$ one has
\begin{align}
	&V_N^{-1/2}(H_N+Y_N-E)V_N^{-1/2}\geq -\frac{1}{2\Delta} V_N^{-1/2}A_NV_N^{-1/2}+\idty+\gamma P_{\SSS_{L,K}^N}V_N^{-1}-EV_N^{-1}\\\geq& \left(1-\frac{2J}{\Delta}\right)+\left(\left(1-\frac{2J}{\Delta}\right)K-E_{K,\delta/2}\right)V_{N}^{-1}P_{\SSS_{L,K}^N}-E_{K,\delta/2}\overline{P}_{\SSS_{L,K}^N}V_N^{-1}\\=& \left(1-\frac{2J}{\Delta}\right)\left(\idty-(1-\delta/2) V_N^{-1}P_{\SSS_{L,K}^N}-(K+1-\delta/2)V_N^{-1}\overline{P}_{\SSS_{L,K}^N}\right)\:.
	\label{eq:estimate}
\end{align}
Now, note that 
\begin{equation}
-(1-\delta/2)V^{-1}_NP_{\SSS_{L,K}^N}\geq -(1-\delta/2)P_{\SSS_{L,K}^N} \quad\mbox{as well as}\quad -V^{-1}_N\overline{P}_{\SSS_{L,K}^N}\geq -\frac{1}{K+1}\overline{P}_{\SSS_{L,K}^N}\:, 
\end{equation}
which we use to further estimate \eqref{eq:estimate}:
\begin{equation}
\eqref{eq:estimate}\geq \left(1-\frac{2J}{\Delta}\right)\left(\frac{\delta}{2} P_{\SSS_{L,K}^N}+\frac{\delta/2}{K+1}\overline{P}_{\SSS_{L,K}^N}\right)\geq \left(1-\frac{2J}{\Delta}\right)\frac{\delta}{2(K+1)}\:.
\end{equation}
From this, it can be concluded that for any $E\in[0,E_{K,\delta}]$, we have $E\notin\sigma(H_N+W_N+\gamma P_{\SSS_{L,K}^N})$ and moreover that
\begin{equation}
\|V_N^{1/2}(H_N+W_N+\gamma P_{\SSS_{L,K}^N}-E)^{-1}V_N^{1/2}\|\leq \frac{2(K+1)}{\delta\left(1-\frac{2J}{\Delta}\right)}
\end{equation}
and, by a slight modification (see \cite[Lemma 4.3]{EKS}), one gets for any $\varepsilon\in\R$:
\begin{equation}
\|V_N^{1/2}(H_N+W_N+\gamma P_{\SSS_{L,K}^N}-E+i\varepsilon)^{-1}V_N^{1/2}\|\leq \frac{4(K+1)}{\delta\left(1-\frac{2J}{\Delta}\right)}\:.
\end{equation}
From Theorem \ref{prop:CT}, we therefore get for any $\mathcal{A}\in\SSS_L^N\setminus\SSS_{L,K}^N$
\begin{equation} \label{eq:CTspinJ}
\|P_{\SSS_{L,K}^N}(H_N+W_N+\gamma P_{\SSS_{L,K}^N}-E+i\varepsilon)^{-1}P_{\mathcal{A}}\|\leq \frac{8(K+1)}{V_{N,0}\cdot\delta\left(1-\frac{2J}{\Delta}\right)}e^{-\mu_{K} d^N(\mathcal{A},\SSS_{L,K}^N)}\:,
\end{equation}
for any $E\in[0,E_{K,\delta/2}]$ and any $\varepsilon\in\R$. The constant $\mu_{K}=\mu_{K}(\Delta,\delta,J)$ is given by
\begin{equation}
\mu_{K}=\log\left(1+\frac{\delta (\Delta-2J)}{16J(K+1)}\right)\:.
\end{equation}
Now, let $\Gamma$ be the circle centered at $\frac{1}{2}E_{K,\delta}$ with radius $R:=\frac{1}{2}\left(1-\frac{2J}{\Delta}\right)(K+1)$. Note that this implies $\dist(\Gamma, [0,E_{K,\delta}])=\frac{\delta}{2}\left(1-\frac{2J}{\Delta}\right)=:\delta'$. Moreover, it follows from Proposition \ref{prop:relativebound} that
\begin{align}
&H_N+W_N+\gamma P_{\SSS_{L,K}^N}\geq-\frac{1}{2\Delta}A_N+V_N+\gamma P_{\SSS_{L,K}^N}\geq \left(1-\frac{2J}{\Delta}\right)V_N+\gamma P_{\SSS_{L,K}^N}\\&\geq \left(1-\frac{2J}{\Delta}\right)(V_N+K)P_{\SSS_{L,K}^N}+\left(1-\frac{2J}{\Delta}\right)V_N\overline{P}_{\SSS_{L,K}^N}\\&\geq \left(1-\frac{2J}{\Delta}\right)(K+1)\:,
\end{align}
which means that there is no spectrum of $H_N+W_N+\gamma P_{\SSS_{L,K}^N}$ inside the circle $\Gamma$, which implies
\begin{equation}
\oint_\Gamma (H_N+W_N+\gamma P_{\SSS_{L,K}^N}-z)^{-1}\text{d}z=0\:.
\end{equation}
We therefore get
\begin{align}
Q_{K,\delta}^N=&(Q_{K,\delta}^N)^2=\frac{i}{2\pi }Q_{K,\delta}^N\oint_\Gamma (H_N+W_N-z)^{-1}\text{d}z\\=&\frac{i}{2\pi }Q_{K,\delta}^N\oint_\Gamma [(H_N+W_N-z)^{-1}-(H_N+W_N+\gamma P_{\SSS_{L,K}^N}-z)^{-1}]\text{d}z\\=&\frac{i\gamma}{2\pi}Q_{K,\delta}^N\oint_\Gamma (H_N+W_N-z)^{-1}P_{\SSS_{L,K}^N}(H_N+W_N+\gamma P_{\SSS_{L,K}^N}-z)^{-1}\text{d}z
\end{align}
We then proceed to estimate 
\begin{align}
	\|Q_{K,\delta}^NP_\mathcal{A}\|&\leq \gamma R\max_{z\in\Gamma}\left[\|Q_{K,\delta}^N(H_N+W_N-z)^{-1}\|\|P_{\SSS_{L,K}^N}(H_N+W_N+\gamma P_{\SSS_{L,K}^N}-z)^{-1}P_\mathcal{A}\|\right]\\
	&\leq \frac{\gamma R}{\delta'}\max_{z\in\Gamma}|\|P_{\SSS_{L,K}^N}(H_N+W_N+\gamma P_{\SSS_{L,K}^N}-z)^{-1}P_\mathcal{A}\|\:,
	\label{eq:specCT}
	\end{align}
where we have used $\|Q_{K,\delta}^N(H_N+W_N-z)^{-1}\|\leq (\delta')^{-1}$. Note that for any $z\in\Gamma$, we have 
\begin{equation}
\Re z\leq\left(1-\frac{2J}{\Delta}\right)\left(K+1-\frac{\delta}{2}\right)=E_{K,\delta/2}\:,
\end{equation}
which means that we can apply \eqref{eq:CTspinJ} to estimate $\|P_{\SSS_{L,K}^N}(H_N+W_N+\gamma P_{\SSS_{L,K}^N}-z)^{-1}P_\mathcal{A}\|$ uniformly in $z\in\Gamma$. We then get
\begin{equation}
\eqref{eq:specCT}\leq \frac{8\gamma R(K+1)}{V_{N,0}\delta'\delta(1-\frac{2J}{\Delta})}e^{-\mu_{K}d^N(\mathcal{A},\SSS_{L,K}^N)}=\frac{8K(K+1)}{V_{N,0}\delta^2} e^{-\mu_{K}d^N(\mathcal{A},\SSS_{L,K}^N)}\:,
\end{equation}
which is the desired result.
\end{proof} 
\begin{remark}
This result applies in particular to non-negative multiplication operators $W_N$, that are of the form
\begin{equation}
W_N(X)=\sum_{i=1}^N \nu(x_i)
\end{equation} 
for any $X=(x_1,x_2,\dots,x_N)\in\SSS_L^N$, where $\nu:\Lambda_L\rightarrow \R^+_0$ is an arbitrary non-negative function with domain $\Lambda_L$. In this case, $W_N$ corresponds to a background magnetic field in $3$-direction, whose value at each site $j$ is given by $\nu(j)$. To be more precise, we have
\begin{equation}
W_N=U_L^N \left(\sum_{j=1}^L\mathcal{N}_j^{loc}\nu(j)\upharpoonright_{\mathcal{H}_L^N}\right)(U_L^N)^*\:,
\end{equation}
where the unitary operator $U_L^N$ was given in \eqref{eq:unitary}.
See also \cite[Remark 2.7]{F}. 
\end{remark}
 \subsection{Entanglement Entropy}
 In what follows, we will mainly use ideas and previous results from \cite[Sec.\ 5]{ARFS}. Let $\SSS_L:=\bigcup_{N=0}^{2JL}\SSS_L^N$, which allows us to identify 
 \begin{equation} \label{eq:hilspace}
 \ell^2(\SSS_L)=\bigoplus_{N=0}^{2JL}\ell^2(\SSS_L^N)\:.
 \end{equation} 
 Next, let $\psi\in\ell^2(\SSS_L)$ be normalized.
We denote the associated density matrix by $\rho_\psi=|\psi\rangle\langle \psi|$.\footnote{Using physicist's notation, for any $\alpha, \beta$ in a Hilbert space $\mathcal{H}$, the symbol ``$|\alpha\rangle\langle\beta|$" denotes the rank-one operator $\psi\mapsto \alpha\langle \beta,\psi\rangle$ for any $\psi\in\mathcal{H}$.}
 Analogous to \eqref{eq:SSS}, for any subset $\Gamma\subset\Lambda_L$, and any $N\in\{1,\dots,2J|\Gamma|\}$, we define
\begin{equation}
\SSS_{\Gamma}^N:=\left\{(x_1,x_2,\dots,x_N)\in\Gamma^N: x_1\leq x_2\leq\cdots\leq x_N\mbox{ and } \min_{k\in\{1,2,\dots, N-2J\}}(x_{k+2J}-x_k)\geq 1 \right\}
\end{equation}
and -- as before -- we set $\SSS_\Gamma^0:=\{\emptyset\}$.
Moreover, $\ell^2(\SSS_\Gamma)$ is defined analogously to \eqref{eq:hilspace}. 
Now, let $\ell\in\N$ such that $1<\ell<L$. Then, $\Lambda_L=\Lambda_\ell\cup\Lambda_\ell^c$ is a spatial bipartition of $\Lambda_L$ into two disjoint discrete subintervals, with the corresponding decomposition of the Hilbert space $\ell^2(\SSS_L)=\ell^2(\SSS_{\Lambda_\ell})\otimes\ell^2(\SSS_{\Lambda_\ell^c})$. Moreover, for any $X\in\SSS_{\Lambda_\ell}$ and any $Y\in\SSS_{\Lambda_\ell^c}$, which implies that $X\cup Y\in\SSS_L$, one naturally identifies $\phi_{X\cup Y}=\phi_X\otimes \phi_Y$.\footnote{For two multisets $X=(x_1,x_2,\dots,x_j)\in\SSS_{\Lambda_\ell}$ and $Y=(y_1,y_2,\dots,y_K)\in\SSS_{\Lambda_\ell^c}$, the notation ``$X\cup Y$" denotes the multiset $(x_1,x_2,\dots,x_j,y_1,y_2,\dots,y_k)\in\SSS_L$. This means in particular that ${\bf m}_{X\cup Y}={\bf m}_X+{\bf m}_Y$.}

The reduced state $\rho_1: \ell^2(\SSS_{\Lambda_\ell})\rightarrow \ell^2(\SSS_{\Lambda_\ell})$ is the linear operator given by
\begin{equation} \label{eq:reducedstate}
\rho_1(\psi,\Lambda_\ell)\equiv\rho_1=\mbox{Tr}_{\Lambda_\ell^c}(\rho_\psi)=\sum_{X_1,X_2\in\SSS_{\Lambda_\ell}}\sum_{Y\in\SSS_{\Lambda_\ell^c}}\psi(X_1\cup Y)\overline{\psi(X_2\cup Y)}|\phi_{X_1}\rangle\langle \phi_{X_2}|
\end{equation}
where $\mbox{Tr}_{\Lambda_\ell^c}(\cdot)$ denotes the partial trace over the subsystem $\Lambda_\ell^c$. 
Then, the entanglement entropy of $\rho_\psi$, which we denote by $\mathcal{E}(\rho_\psi)$ is given by 
\begin{equation}
\mathcal{E}(\rho_\psi)=-\mbox{Tr}(\rho_1\log\rho_1)=:\mathcal{S}(\rho_1)\:,
\end{equation}
where $\mathcal{S}(\rho_1)$ denotes the von Neumann entropy of the reduced state $\rho_1$.\footnote{We adapt the convention $0\log 0=0$.} As in \cite{ ARFS, BW18}, we will actually show estimates for the $\alpha$-R\'enyi entropies $\mathcal{S}_\alpha$ of $\rho_1$, which are given by
\begin{equation}
\mathcal{S}_\alpha(\rho_1):=\frac{1}{1-\alpha}\log\mbox{Tr}[(\rho_1)^\alpha]\:,
\end{equation}
where $\alpha\in(0,1)$. Since for every $\alpha\in(0,1)$ one has $\mathcal{S}(\rho_1)\leq \mathcal{S}_\alpha(\rho_1)$, showing suitable bounds for $\mathcal{S}_\alpha(\rho_1)$ will then readily imply the desired result for the entanglement entropy $\mathcal{E}(\rho_\psi)$. Let us now further analyze \eqref{eq:reducedstate}:
\begin{equation} \label{eq:reducedstatesplit}
\rho_1=\sum_{X_1,X_2\in\SSS_{\Lambda_\ell}}\psi(X_1)\overline{\psi(X_2)}|\phi_{X_1}\rangle\langle\phi_{X_2}|+\sum_{X_1,X_2\in\SSS_{\Lambda_\ell}}\sum_{Y\in\SSS_{\Lambda_\ell^c}\setminus\{\emptyset\}}\psi(X_1\cup Y)\overline{\psi(X_2\cup Y)}|\phi_{X_1}\rangle\langle\phi_{X_2}|\:,
\end{equation}
where the first sum simply corresponds to the contributions in \eqref{eq:reducedstate}, for which $Y=\emptyset$. Now, let the vector $\Psi\in\ell^2(\SSS_{\Lambda_\ell})$ be given by $\Psi:=\sum_{X_3\in\SSS_{\Lambda_\ell}}\psi(X_3)\phi_{X_3}$, which clearly satisfies $\|\Psi\|_{\ell^2(\SSS_{\Lambda_\ell})}\leq\|\psi\|_{\ell^2(\SSS_L)}$ and  observe that
\begin{equation}
\sum_{X_1,X_2\in\SSS_{\Lambda_\ell}}\psi(X_1)\overline{\psi(X_2)}|\phi_{X_1}\rangle\langle\phi_{X_2}|=|\Psi\rangle\langle \Psi|\:.
\end{equation}
So, we write $\rho_1=|\Psi\rangle\langle\Psi|+\hat{\rho}_1$ with $\hat{\rho}_1$ being equal to the second sum in \eqref{eq:reducedstatesplit}.
Let us now focus on $\mbox{Tr}[(\rho_1)^\alpha]$:
\begin{align}
\mbox{Tr}[(\rho_1)^\alpha]&=\mbox{Tr}[(|\Psi\rangle\langle\Psi|+\hat{\rho}_1)^\alpha]\leq 2\mbox{Tr}[(|\Psi\rangle\langle\Psi|)^\alpha]+2\mbox{Tr}[(\hat{\rho}_1)^\alpha]\leq 2+2\mbox{Tr}[(\hat{\rho}_1)^\alpha]\:,
\end{align}
where we have used the quasi-norm property of $\mbox{Tr}|\cdot|^\alpha$, cf.\ \cite[Satz 3.21]{Weidmann} and the fact that $|\Psi\rangle\langle \Psi|$ is a non-negative rank-one operator with norm less than or equal to one.
Let us now further estimate 
\begin{align}\label{eq:teileins}
\mbox{Tr}[(\hat{\rho}_1)^\alpha]&=\sum_{X\in\SSS_{\Lambda_\ell}}\langle \phi_X,(\hat{\rho}_1)^\alpha\phi_X\rangle\leq \sum_{X\in\SSS_{\Lambda_\ell}}\langle \phi_X,\hat{\rho}_1\phi_X\rangle^\alpha=\sum_{j=0}^{2J\ell}\sum_{X\in\SSS_{\Lambda_\ell}^j}\langle\phi_X,\hat{\rho}_1\phi_X\rangle^\alpha
\\&\leq \left|\bigcup_{j=0}^{4J-1}\SSS_{\Lambda_\ell}^j\right|^{1-\alpha}+\sum_{j=4J}^{2J\ell}\sum_{X\in\SSS_{\Lambda_\ell}^j}\langle\phi_X,\hat{\rho}_1\phi_X\rangle^\alpha \label{eq:teilzwei}
\end{align}
where we have used Jensen's inequality for the first estimate in \eqref{eq:teileins} and the second estimate \eqref{eq:teilzwei} just follows from maximizing $\sum\sum \langle \phi_X,\hat{\rho}_1\phi_X\rangle^\alpha$ under the constraint $\sum\sum\langle\phi_X,\hat{\rho}_1\phi_X\rangle\leq 1$. We estimate further by noting $|\SSS_{\Lambda_\ell}^j|\leq \ell^j$ and thus -- assuming $\ell\geq 4J$ --
\begin{equation}
\eqref{eq:teilzwei}\leq \ell^{4J(1-\alpha)}+\sum_{j=4J}^{2J\ell}\sum_{X\in\SSS_{\Lambda_\ell}^j}\langle\phi_X,\hat{\rho}_1\phi_X\rangle^\alpha\:.
\end{equation}
Now, for each $X\in\SSS_{\Lambda_\ell}$ with $|X|=j$, we get
\begin{align}
\langle \phi_X,\hat{\rho}_1\phi_X\rangle&=\sum_{Y\in\SSS_{\Lambda_\ell^c}\setminus\{\emptyset\}}|\psi(X\cup Y)|^2=\|P_{\mathcal{A}_X}\psi\|^2\:,
\end{align}
where given any $X\in\SSS_{\Lambda_\ell}$, we have defined $\mathcal{A}_{X}:=\{X\cup Y: Y\in\SSS_{\Lambda_\ell^c}\setminus\{\emptyset\}\}$.
Altogether, these considerations show the following 
\begin{lemma} Let $\ell\geq 4J$ and  $\psi\in\ell^2(\SSS_L)$ be normalized. Then, for any $\alpha\in(0,1)$ we get
	\begin{equation}
		\mbox{\emph{Tr}}[({\rho}_1)^\alpha]\leq  2+2\ell^{4J(1-\alpha)}+2\sum_{j=4J}^{2J\ell}\sum_{X\in\SSS_{\Lambda_\ell}^j}\|P_{\mathcal{A}_{X}}\psi\|^{2\alpha}\:.
	\end{equation}
	\label{lemma:traceestimate}
\end{lemma}

\section{Proof of the logarithmically corrected entanglement bound} \label{sec:3}

\subsection{Local distance estimates}
The following lemma, whose proof can be found in Appendix \ref{app:b2}, provides us with an estimate for the distance of a configuration in $\SSS_{L}$ to the nearest configuration in $\SSS_{L,K}$. 

\begin{lemma} \label{lemma:localminimizer}
Let $X=(x_1,x_2,\dots,x_j)\in\SSS_{\Lambda_\ell}^j$ with $j\geq 4J$ and $Y=(y_1,y_2,\dots,y_k)\in\SSS_{\Lambda_{\ell}^c}^k$, where $k\in\N$. Then, for any $K\geq 4J^2$, we have the following estimate:
\begin{equation}d^j(X,\SSS_{\Lambda_\ell,K}^j)\leq d^{j+k}(X\cup Y,\SSS_{L,K}^{j+k})\:. \label{eq:distest1}
\end{equation}
\end{lemma}
So, the main observation made in Lemma \ref{lemma:localminimizer} is that there always exists a configuration in $\SSS^j_{\Lambda_\ell,K}$ which is at least as close to $X$ as any configuration in $\SSS^{j+k}_{L,K}$ might be to $X\cup Y$. 
We now combine Lemmas \ref{lemma:traceestimate}, \ref{lemma:localminimizer} and Theorem \ref{thm:specestimate}. To this end, for any non-negative function $W: \SSS_L\rightarrow \R_0^+$, let us define $Q_{K,\delta}\equiv Q_{K,\delta}(L,W):=\bigoplus_{N=0}^{2JL}Q_{K,\delta}^N(L,W_N)$ -- the spectral projection of the full Hamiltonian onto the energy interval $[0,E_{K,\delta}]$ with background potential $W$. Here, $W_N:\ell^2(\SSS_L^N)\rightarrow\ell^2(\SSS_L^N)$, denotes the multiplication operator induced by $W$, i.e. $(W_Nf)(X):=W(X)f(X)$, for any $f\in\ell^2(\SSS_L^N)$ and any $X\in\SSS_L^N$.
\begin{cor} \label{cor:traceestimate2}
Fix $K\in\N$ and let $\ell\geq 4J$. Let $\psi\in \mbox{\emph{ran}}(Q_{K,\delta})$ be normalized. Then, for any $\alpha\in(0,1)$, we have the following estimate
	\begin{equation} \label{eq:trace2}
	\mbox{\emph {Tr}}[(\rho_1)^\alpha]\leq 2+2\ell^{4J(1-\alpha)}+2C_K'^{2\alpha}\sum_{j=4J}^{2J\ell}\sum_{X\in\SSS_{\Lambda_\ell}^j}e^{-2\alpha\mu_Kd^j(X,\SSS^j_{\Lambda_\ell,K})}\:,
	\end{equation}
	where
	\begin{equation}
	C_K'=\max\left\{1,\frac{2K(K+1)}{J^2\delta^2}\right\}\:.
	\end{equation}
\end{cor}
\begin{proof}
	Let $X\in \SSS_{\Lambda_\ell}^j$, where $j\geq 4J$. Defining $\mathcal{A}_{X,k}=\{X\cup Y: Y\in\SSS^k_{\Lambda_\ell^c}\}$ and using that
	\begin{equation}
	P_{\mathcal{A}_X}Q_{K,\delta}=\bigoplus_{k=1}^{2J(L-\ell)}P_{\mathcal{A}_{X,k}}Q_{K,\delta}^{j+k}
	\end{equation}
we get 
	\begin{align}
	\|P_{\mathcal{A}_X}\psi\|&=\|P_{\mathcal{A}_X}Q_{K,\delta}\psi\|\leq \|P_{\mathcal{A}_X}Q_{K,\delta}\|= \max_{k\in\{1,\dots, 2J(L-\ell)\}}\|P_{\mathcal{A}_{X,k}}Q_{K,\delta}^{j+k}\|\\&\leq \max_{k\in\{1,\dots,2J(L-\ell)\}} \left(C_{j+k,K}\cdot e^{-\mu_K d^{j+k}(\mathcal{A}_{X,k},\SSS_{L,K}^{j+k})}\right)
	\label{eq:findmax}
	\end{align}
	where we have used Theorem \ref{thm:specestimate} for the last inequality and the constants $C_{j+k,K}$ and $\mu_K$ were given in \eqref{eq:fullCT2'2}. Now, since $j+k\geq 4J$ by assumption, it follows from Proposition \ref{prop:minimizer} that $V_{j+k,0}=4J^2$ and thus 
	\begin{equation}
	C_{j+k,K}=\max\left\{1,\frac{2K(K+1)}{J^2\delta^2}\right\}=C'_K\:,
	\end{equation}
	independently of $j$. Now, for any $k\in\{1,\dots, 2J(L-\ell)\}$, observe that by Lemma \ref{lemma:localminimizer}, we have
	\begin{equation}
	d^{j+k}(\mathcal{A}_{X,k},\SSS_{L,K}^{j+k})=\min_{Y\in\SSS_{\Lambda_\ell^c}^k} d^{j+k}(X\cup Y, \SSS_{L,K}^{j+k})\geq d^j(X,\SSS_{\Lambda_\ell,K}^j)
	\end{equation}
	and thus we may proceed to estimate \eqref{eq:findmax} to get
	\begin{equation}
	\|P_{\mathcal{A}_X}\psi\|\leq \eqref{eq:findmax}\leq C'_K\cdot e^{-\mu_K d^j(X,\SSS^j_{\Lambda_\ell,K})}\:.
	\end{equation}
	Using Lemma \ref{lemma:traceestimate}, we consequently find
	\begin{equation}
	\mbox{Tr}[(\rho_1)^\alpha]\leq 2+2\ell^{4J(1-\alpha)}+2C_K'^{2\alpha}\sum_{j=4J}^{2J\ell}\sum_{X\in\SSS_{\Lambda_\ell}^j}e^{-2\alpha\mu_Kd^j(X,\SSS^j_{\Lambda_\ell,K})}\:,
	\end{equation}
	which shows \eqref{eq:trace2}. 
\end{proof}
\subsection{Estimates using geometric series}
\begin{lemma}\label{lemma:geom1d}
Let $\mathcal{X}^N= \set{X=(x_1,x_2,\ldots,x_N)\in \Z^N:x_1<x_2<\cdots<x_N}$ be the set of spin--1/2
configurations on the infinite chain. Moreover, let $C=(c_1,c_2,\ldots,c_N)= (c, c+1, \ldots, c+(N-1))$ be an arbitrary configuration of $N$ consecutive particles ($c\in \Z$). Then, for any $\gamma>0$, we get
\begin{equation}
    \sum_{X\in\mathcal{X}^N}e^{-\gamma d^N(X,C)}\leq \left(\frac{1}{1-e^{-\gamma}}\right)\left(\prod_{k=1}^\infty\frac{1}{1-e^{-k\gamma}}\right)^2 := \mathcal{L}_\gamma<\infty\:.
\end{equation}
\begin{proof}

Define 
$$\mathcal{X}^{N,0} = \set{X\in\mathcal{X}^N: x_1\geq c_1},\;\;\;\;\mathcal{X}^{N,N}=\set{X\in \mathcal{X}^N: x_N< c_N}$$
and for any $j\in \set{1,\ldots, N-1}$,
\[\mathcal{X}^{N,j}=\set{X\in \mathcal{X}^N: x_j<c_j,\; x_{j+1}\geq c_{j+1}}\:.  \]
Then, $\displaystyle\mathcal{X}^N= \uplus_{j=0}^N\mathcal{X}^{N,j}$, where ``$\uplus$" denotes a disjoint union. For any $j\in \set{1,\ldots, N-1}$, we get by an argument similar to \cite[Lemma A.3]{FSch} that
\begin{equation}\begin{split}
    \sum_{X\in \mathcal{X}^{N,j}}e^{-\gamma d^N(X,C)}&=\sum_{x_1<x_2<\cdots<x_j<c_j}e^{-\gamma(|x_1-c_1|+\cdots+|x_j-c_j|)}  \sum_{c_{j+1}\leq x_{j+1}<x_{j+2}<\cdots<x_N}e^{-\gamma(|x_{j+1}-c_{j+1}|+\cdots+|x_N-c_N|)}\\
    &\leq e^{-\gamma j}\left(\prod_{k=1}^j\frac{1}{1-e^{-k\gamma}}\right)\left(\prod_{k=1}^{N-j}\frac{1}{1-e^{-k\gamma}}\right)
    \leq e^{-\gamma j}\left(\prod_{k=1}^\infty\frac{1}{1-e^{-k\gamma}}\right)^2\:.\\
\end{split}
\end{equation}
Analogously, we find
\begin{equation}
    \sum_{X\in\mathcal{X}^{N,0}}e^{-\gamma d^N(X,C)}\leq \prod_{k=1}^N\frac{1}{1-e^{-k\gamma}}\leq \left(\prod_{k=1}^\infty \frac{1}{1-e^{-k\gamma}}\right)^2
\end{equation}
\begin{equation}
    \sum_{X\in\mathcal{X}^{N,N}}e^{-\gamma d^N(X,C)}\leq e^{-\gamma N}\prod_{k=1}^N\frac{1}{1-e^{-k\gamma}}\leq e^{-\gamma N}\left(\prod_{k=1}^\infty \frac{1}{1-e^{-k\gamma}}\right)^2\:.
\end{equation}
Hence,
\begin{align}
    \sum_{X\in\mathcal{X}^N}e^{-\gamma d^N(X,C)}=&\sum_{j=0}^{N}\sum_{X\in \mathcal{X}^{N,j}}e^{-\gamma d^N(X,C)}\\\leq&\sum_{j=0}^Ne^{-\gamma j}\left(\prod_{k=1}^\infty\frac{1}{1-e^{-k\gamma}}\right)^2\leq\left(\frac{1}{1-e^{-\gamma}}\right)\left(\prod_{k=1}^\infty\frac{1}{1-e^{-k\gamma}}\right)^2\:,
\end{align}
which is the desired result. The infinite product's convergence follows from elementary facts.
\end{proof}
\end{lemma}

\begin{defn} \label{def:buildingblocks}
	For any ${\bf m}\in{\bf M}_L^N$, let $\supp({\bf m}):=\{j\in\{1,2,\dots,L\}:{\bf m}(j)\neq 0\}$ denote the support of ${\bf m}$. Then, for any $L$, we define the set ${\bf K}_L$ of ``building bocks" in $\Lambda_L$ as follows
	\begin{align} \label{eq:buildingblocks}
		{\bf K}_L&:=\{{\bf m}\in{\bf M}_L: |\mbox{supp}({\bf m})|=1\}\\&\cup \{{\bf m}\in{\bf M}_L: \exists \alpha,\beta\in\Lambda_L: {\bf m}(i)=2J \mbox{ if } \alpha\leq i\leq \beta\mbox{ and } {\bf m}(i)=0\mbox{ else} \}\:.
	\end{align}
	In other words, building blocks are configurations of either a single column of up to $2J$ particles or of discrete intervals, where each site is occupied by exactly $2J$ particles (rectangular blocks).

\end{defn}

\begin{lemma} \label{lemma:geom2J}
Let $\gamma>0$. Then for any $C\in {\bf K}_L$, we have
\begin{equation}
    \sum_{X\in\SSS^N_L}e^{-\gamma d^N(X,C)}\leq \mathcal{L}^{2J}_\gamma
\end{equation}
\begin{proof}
Let $C$ be a rectangular block, i.e.\ $\;C=(\stackrel{2J\;\;\text{times}}{\overbrace{c,\ldots,c}},\stackrel{2J \;\;\text{times}}{\overbrace{c+1,\ldots,c+1}},\ldots,\stackrel{2J \;\;\text{times}}{\overbrace{c+m,\ldots,c+m}})$ where $N=2J(m+1)$. For any $N$-particle configuration $X=(x_1,\ldots,x_N)\in \SSS^N_L$, let $Z_0(X)$ denote the constraint 
\begin{equation}
    Z_0(X) : x_1\leq x_2\leq\cdots\leq x_N\:.
\end{equation}
Moreover, define the additional constraints
\begin{equation}
Z_i(X): x_i<x_{i+2J}<x_{i+4J}<\cdots<x_{i+2JR_i}
\end{equation} 
where $i\in \set{1,2,\ldots, 2J}$ and $R_i=\max\set{R\in\{1,2,\dots,N\}: i+2JR\leq N}$.

Let 
\begin{equation}\M^N_L:=\set{X\in\Lambda^N_L:\;Z_i(X) \;\text{ holds }\;\forall\; i=0,\ldots,2J }
\end{equation}
and 
\begin{equation}\Gamma^N:=\set{X\in \Z^N:\; Z_k(X)\;\text{ holds }\; \forall\;k=1,\ldots, 2J}\:.
\end{equation}
Obviously, $\M^N_L\subseteq \Gamma^N$. To see that $\M^N_L=\SSS^N_L$ (defined in (\ref{eq:SSS})), note that it is obvious that $\M^N_L\subseteq \SSS^N_L $ since for any $X\in \M^N_L$, the constraints $Z_i(X)$ imply that $\min\set{(x_{k+2J}-x_k): k=1,\ldots, N-2J}\geq 1$.

Conversely, suppose there exists $X\in \SSS^N_L\setminus \M^N_L$, i.e.\ $X$ violates at least one constraint, $Z_t(X)$ say, where $t\neq 0$. Then there exists $r\in \N$ such that $x_{t+2Jr}\geq x_{t+2J(r+1)}$. This is a contradiction since $X\in \SSS^N_L$. Hence, $\SSS^N_L=\M^N_L$.

Therefore, by Lemma (\ref{lemma:geom1d}),
\begin{equation}
    \begin{split}
        \sum_{X\in \SSS^N_L}e^{-\gamma d^N(X,C)}&\leq \sum_{X\in \Gamma^N} e^{-\gamma d^N(X,C)}\\
        &=\sum_{Z_1(X)}e^{\displaystyle-\gamma(|x_1-c_1|+\cdots+|x_{1+2JR_1}-c_{1+2JR_1}|)}\times\cdots\\
        &\;\;\;\;\;\;\;\;\;\;\;\;\;\;\;\;\;\;\;\;\;\;\;\;\;\;\;\;\;\cdots\times\sum_{Z_{2J}(X)}e^{\displaystyle-\gamma(|x_{2J}-c_{2J}|+\cdots+ | x_{2J+2JR_{_{2J}}}- c_{2J+2JR_{_{2J}}}|)}\\
        &\leq \mathcal{L}^{2J}_\gamma\:.
    \end{split}
\end{equation}
If $C=(c_1,\ldots,c_N)$ is a column block, i.e.\ $c_1=\cdots=c_N=c$, where $c\in \Lambda_L$ and $N\leq 2J$, then
\begin{align}
	\sum_{X\in \SSS^N_L}e^{-\gamma d^N(X,C)}\leq &\sum_{-\infty<x_1\leq x_2\leq\cdots\leq x_N<\infty}e^{-\gamma [|x_1-c|+\cdots+|x_N-c|]}\\ \leq& \sum_{x_1, x_2,\cdots, x_N\in\Z}e^{-\gamma [|x_1-c|+\cdots+|x_N-c|]} =\left(\frac{1+e^{-\gamma}}{1-e^{-\gamma}}\right)^N\leq\mathcal{L}^N_{\gamma}\leq\mathcal{L}^{2J}_\gamma\:,
\end{align}
which finishes the proof.
\end{proof}
\end{lemma}
\begin{remark}
Note that for this result, we are making use of the generalized Pauli principle, which requires that no site be occupied by more than $2J$ particles. 
\end{remark}
\subsection{Building blocks and potential} 
Given a value $K\in\N$ and a particle number $N\geq 4J$, it seems rather cumbersome to give a full description of all configurations in ${\bf M}_{L,K}^N$. However, it can be shown that any configuration in ${\bf M}_{L,K}^N$ can be composed out of no more than $K$ of the ``building blocks" as described in Definition \ref{def:buildingblocks} before.
\begin{defn} Given a configuration ${\bf m}\in{\bf M}_L$, we define the quantity $B({\bf m})$ as follows:
	\begin{equation} \label{eq:defB}
	B({\bf m}):=\left|\{i\in\{1,2,\dots,L-1\}:{\bf m}(i)+{\bf m}(i+1)\notin\{0,4J\}\}\right|+(2-\delta_{{\bf m}(1),0}-\delta_{{\bf m}(L),0})\:.
	\end{equation}
	Moreover, for any $R\in\N$, we define
	\begin{equation}
	{\bf B}_{L,R}:=\{{\bf m}\in{\bf M}_L: B({\bf m})\leq R\}\quad\mbox{as well as} \quad {\bf B}_{L,R}^N:={\bf B}_{L,R}\cap {\bf M}_L^N
	\end{equation}
	for any $N\in\{0,1,\dots,2JL\}$. Additionally, for multisets, we introduce $\mathbb{B}_{L,R}=\{X\in\SSS_L : {\bf m}_X\in {\bf B}_{L,R}\}$ and $\mathbb{B}_{L,R}^N=\mathbb{B}_{L,R}\cap\SSS_L^N$.
\end{defn}	
\begin{remark}
The purpose of $B({\bf m})$ is to count the number of neighboring sites $\{i,i+1\}$, $i=1,2,\dots,L-1$, which are not both occupied by either $0$ or $2J$ particles. Each of the two additional terms in $(2-\delta_{{\bf m}(1),0}-\delta_{{\bf m}(L),0})=(1-\delta_{{\bf m}(1),0})+(1-\delta_{{\bf m}(L),0})$ increases $B({\bf m})$ by a value of one if the sites $1$ or $L$ are occupied by any particles (which formally corresponds to including the edges $\{0,1\}$ and $\{L,L+1\}$ in the count). 
\end{remark}
\begin{remark}
Observe that $B({\bf m})\leq J^{-1}V({\bf m})$ for every ${\bf m}\in{\bf M}_L$. This follows from the fact that $v({\bf m}(i), {\bf m}(i+1))=0$ if and only if ${\bf m}(i)={\bf m}(i+1)=0$ or ${\bf m}(i)={\bf m}(i+1)=2J$ and $v({\bf m}(i),{\bf m}(i+1))\geq J$ else, and comparing the extra term $(2-\delta_{{\bf m}(1),0}-\delta_{{\bf m}(L),0})$ in \eqref{eq:defB}  with the boundary field term $J({\bf m}(1)+{\bf m}(L))$ in \eqref{eq:potential}. For $K\geq 4J^2$, define $\widetilde{K}:=\lfloor K/J\rfloor$, and observe that this implies
\begin{equation} \label{eq:configincl}
{\SSS}_{L,K}\subset {\mathbb{B}}_{L,\widetilde{K}}\quad\mbox{and thus in particular}\quad {\SSS}^N_{L,{K}}\subset {\mathbb{B}}^N_{L,\widetilde{K}}\:.
\end{equation}
\end{remark}
It is now crucial to observe that any configuration in ${\bf B}_{L,\widetilde{K}}$ can always by obtained by composing it out of at most $(\widetilde{K}-1)$ ``building blocks". 

\begin{remark} \label{rem:buildingblocks}
Observe that for any ${\bf m}\in{\bf M}_L$, there exist $\{{\bf k}^{(i)}\}_{i=1}^\tau\subset{\bf K}_L$, ($\tau\leq L$), with pairwise disjoint supports, such that ${\bf m}=\sum_{i=1}^\tau{\bf k}_i$. If in addition, we have ${\bf m}\in{\bf B}_{L,\widetilde{K}}$, it follows from \eqref{eq:defB} that $\tau\leq (\widetilde{K}-1)$, i.e.\ any configuration in ${\bf B}_{L,\widetilde{K}}$ can be composed out of no more than $(\widetilde{K}-1)$ building blocks. See Figure \ref{fig:generalform} for a pictorial representation of building blocks.
\end{remark}

\begin{figure}[ht]
\centering
\resizebox{8cm}{6cm}
{\begin{tikzpicture}[line cap=round,line join=round,>=triangle 45,x=1cm,y=1cm]
		\begin{axis}[
			x=1cm,y=1cm,
			axis x line=bottom,
			axis y line=none,
			xmajorgrids=true,
			xmin=0.2,
			xmax=7.5,
			ymin=-0.5,
			ymax=6.5,
			]
			\clip(0,0) rectangle (11.32,8.7);
			\draw [line width=2pt] (0,1)-- (8,1);
			\begin{scriptsize}
				\draw [fill=pink] (1,1) circle (2.5pt);
				\draw [fill=pink] (1,1.4) circle (2.5pt);
				\draw [fill=pink] (1,1.8) circle (2.5pt);
				\draw [fill=pink] (1,2.2) circle (2.5pt);
				\draw [fill=green] (3,1) circle (2.5pt);
				\draw [fill=green] (3,1.4) circle (2.5pt);
				\draw [fill=qqqqff] (4,1) circle (2.5pt);
				\draw [fill=qqqqff] (4,1.4) circle (2.5pt);
				\draw [fill=qqqqff] (4,1.8) circle (2.5pt);
				\draw [fill=qqqqff] (4,2.2) circle (2.5pt);
				\draw [fill=qqqqff] (4,2.6) circle (2.5pt);
				\draw [fill=qqqqff] (4,3) circle (2.5pt);
				\draw [fill=qqqqff] (4,3.4) circle (2.5pt);
				\draw [fill=qqqqff] (4,3.8) circle (2.5pt);
				\draw [fill=qqqqff] (4,4.2) circle (2.5pt);
				\draw [fill=qqqqff] (5,1) circle (2.5pt);
				\draw [fill=qqqqff] (5,1.4) circle (2.5pt);
				\draw [fill=qqqqff] (5,1.8) circle (2.5pt);
				\draw [fill=qqqqff] (5,2.2) circle (2.5pt);
				\draw [fill=qqqqff] (5,2.6) circle (2.5pt);
				\draw [fill=qqqqff] (5,3) circle (2.5pt);
				\draw [fill=qqqqff] (5,3.4) circle (2.5pt);
				\draw [fill=qqqqff] (5,3.8) circle (2.5pt);
				\draw [fill=qqqqff] (5,4.2) circle (2.5pt);
				\draw [fill=qqqqff] (6,1) circle (2.5pt);
				\draw [fill=qqqqff] (6,1.4) circle (2.5pt);
				\draw [fill=qqqqff] (6,1.8) circle (2.5pt);
				\draw [fill=qqqqff] (6,2.2) circle (2.5pt);
				\draw [fill=qqqqff] (6,2.6) circle (2.5pt);
				\draw [fill=qqqqff] (6,3) circle (2.5pt);
				\draw [fill=qqqqff] (6,3.4) circle (2.5pt);
				\draw [fill=qqqqff] (6,3.8) circle (2.5pt);
				\draw [fill=qqqqff] (6,4.2) circle (2.5pt);
				\draw [fill=red] (7,1) circle (2.5pt);
				\draw [fill=red] (7,1.4) circle (2.5pt);
				\draw [fill=red] (7,1.8) circle (2.5pt);
				\draw [fill=red] (7,2.2) circle (2.5pt);
				
			\end{scriptsize}
		\end{axis}
	\end{tikzpicture}
}
\caption{A configuration ${\bf m}\in{\bf M}_{L}^N$ of $N=37$ particles (here: $J=9/2$). Moreover, $B({\bf m})=6\leq J^{-1}V({\bf m})=26$ and indeed, ${\bf m}$ can be written as a composition of $4\leq (\lfloor J^{-1}V({\bf m})\rfloor-1)$ building blocks (represented here by four different colors).}
\label{fig:generalform}
\end{figure}
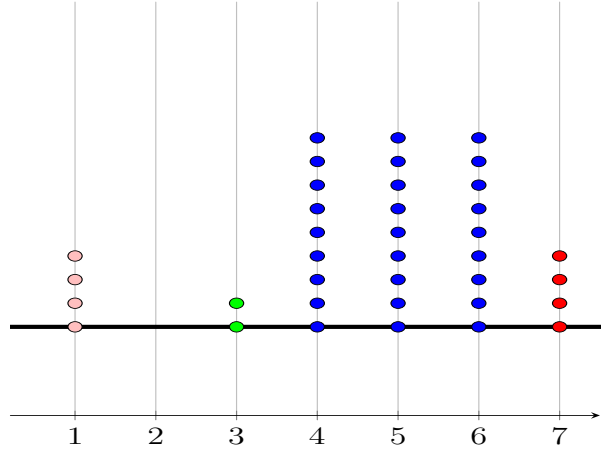
\begin{lemma} Let $j\in\{4J, 4J+1,\dots,2J\ell\}$, $K\in\{4J^2,4J^2+1,\dots, 2J\ell\}$ and $\widetilde{K}:=\lfloor K/J\rfloor$. Then, for any $\gamma>0$, we get the following estimate
	\begin{equation}
	\sum_{X\in\SSS_{\Lambda_\ell}^j}e^{-\gamma d^j(X,\SSS_{\Lambda_\ell,K}^j)}\leq (4Je)^{\widetilde{K}-2}\mathcal{L}_{2\alpha\mu_K}^{2J(\widetilde{K}-1)}\ell^{2\widetilde{K}-3}\:.
	\end{equation}
	\label{lemma:series}
\end{lemma}
\begin{proof}
Analogously to before, we define $\mathbb{B}_{\Lambda_\ell,\widetilde{K}}^j:=\mathbb{B}_{L,\widetilde{K}}^j\cap\SSS_{\Lambda_\ell}$ and observe that due to \eqref{eq:configincl}, we have $\SSS_{\Lambda_\ell,K}^j\subset \mathbb{B}_{\Lambda_\ell,\widetilde{K}}^j$ and thus, we get
\begin{equation} \label{eq:donau}
\sum_{X\in\SSS_{\Lambda_\ell}^j}e^{-\gamma d^j(X,\SSS_{\Lambda_\ell,K}^j)}\leq \sum_{X\in\SSS_{\Lambda_\ell}^j}\sum_{Y\in\SSS_{\Lambda_\ell,K}^j}e^{-\gamma d^j(X,Y)}=\sum_{Y\in\SSS_{\Lambda_\ell,K}^j}\sum_{X\in\SSS_{\Lambda_\ell}^j}e^{-\gamma d^j(X,Y)}\leq \sum_{Y\in\mathbb{B}_{\Lambda_\ell,\widetilde{K}}^j}\sum_{X\in\SSS_{\Lambda_\ell}^j}e^{-\gamma d^j(X,Y)}\:.
\end{equation}
We now claim that for every $Y\in\mathbb{B}_{\Lambda_\ell,\widetilde{K}}^j$, we get
\begin{equation} \label{eq:estimate2JK}
\sum_{X\in\SSS_{\Lambda_\ell}^j}e^{-\gamma d^j(X,Y)}\leq \mathcal{L}_\gamma^{2J(\widetilde{K}-1)}\:.
\end{equation}
Since $Y\in\mathbb{B}_{\Lambda_\ell,\widetilde{K}}^j$, observe that by Remark \ref{rem:buildingblocks}, there exist ${\bf k}^{(1)}, {\bf k}^{(2)}, \dots, {\bf k}^{(\tau)}\in{\bf K}_L$ with $\tau\leq (\widetilde{K}-1)$ and pairwise disjoint support such that 
\begin{equation}
{\bf k}^{(1)}+{\bf k}^{(2)}+\cdots+{\bf k}^{(\tau)}={\bf m}_Y\:.
\end{equation}
Without loss of generality, we may assume that the building blocks are ordered such that $i<j$ implies $\max\supp({\bf k}^{(i)})<\min\supp({\bf k}^{(j)})$.

For any $i\in\{1,2\dots,\tau\}$, let $K^{(i)}:=X_{{\bf k}^{(i)}}$ denote the multiset associated with the building block ${\bf k}^{(i)}$ and thus $Y=K^{(1)}\cup K^{(2)}\cup\dots\cup K^{(\tau)}$. Now, decompose any $X\in\SSS_{\Lambda_\ell}^j$ analogously into $X=X^{(1)}\cup X^{(2)}\cup\dots\cup X^{(\tau)}$ such that $|X^{(i)}|=|K^{(i)}|=:k_i$ for every $i\in\{1,2,\dots,\tau\}$ and $i<j$ implies that $\max (X^{(i)})<\min(X^{(j)})$.

We then get
\begin{align}
\sum_{X\in\SSS_{\Lambda_\ell}^j}e^{-\gamma d^j(X,Y)}&=\sum_{X^{(1)}\cup\dots\cup X^{(\tau)}\in\SSS_{\Lambda_\ell}^j}e^{-\gamma\left(d^{k_1}(X^{(1)},K^{(1)})+\cdots+d^{k_\tau}(X^{(\tau)},K^{(\tau)})\right)}\\&\leq \prod_{i=1}^\tau\left(\sum_{X^{(i)}\in\SSS_{\Lambda_\ell}^{k_i}}e^{-\gamma d^{k_i}(X^{(i)},K^{(i)})}\right)\leq \mathcal{L}_{\gamma}^{2J\tau}\leq \mathcal{L}_\gamma^{2J(\widetilde{K}-1)}\:, \label{eq:estimsum}
\end{align}
where we have used Lemma \ref{lemma:geom2J} for estimating the sum in \eqref{eq:estimsum}. This shows \eqref{eq:estimate2JK} which together with \eqref{eq:donau} proves that
\begin{equation}
\sum_{X\in\SSS_{\Lambda_\ell}^j}e^{-\gamma d^j(X,\SSS_{\Lambda_\ell,K}^j)}\leq \sum_{Y\in{\mathbb{B}_{\Lambda_\ell,\widetilde{K}}^j}} \mathcal{L}_\gamma^{2J(\widetilde{K}-1)}=|\mathbb{B}_{{\Lambda_\ell,\widetilde{K}}}^j| \cdot\mathcal{L}_\gamma^{2J(\widetilde{K}-1)}\:,
\end{equation}
which means that we need to further estimate the number of configurations in $\Lambda_\ell$ of $(\widetilde{K}-1)$ or less building blocks of $j$ particles. We do this by a rather rough combinatorial argument: firstly, note that distributing $j$ particles into up to $(\widetilde{K}-1)$ building blocks can be estimated by $\binom{j+\widetilde{K}-2}{\widetilde{K}-2}$. This is clearly an overestimate, since it disregards the constraints in \eqref{eq:buildingblocks} which building blocks have to satisfy. Next, we have to account for all the possible ways, those $(\widetilde{K}-1)$ or less building blocks can be placed in $\Lambda_\ell$. A trivial upper bound for this is given by $\ell^{\widetilde{K}-1}$, since there are up to $\ell$ sites one could place each individual building block (disregarding the fact that the supports of the building blocks have to be disjoint and can be larger than one and thus again overestimating).
If $\widetilde{K}>2$, we therefore conclude
\begin{equation} \label{eq:setB}
|\mathbb{B}_{{\Lambda_\ell,\widetilde{K}}}^j| \leq \binom{j+\widetilde{K}-2}{\widetilde{K}-2}\cdot\ell^{\widetilde{K}-1}\leq \left(\frac{(2J\ell+\widetilde{K}-2)e}{\widetilde{K}-2}\right)^{\widetilde{K}-2}\ell^{\widetilde{K}-1}\leq(4Je)^{\widetilde{K}-2}\ell^{2\widetilde{K}-3}\:,
\end{equation}
where we have used the estimate $\binom{\alpha}{\beta}\leq \left(\frac{\alpha e}{\beta}\right)^\beta$ for the binomial coefficient as well as the fact that $j\leq 2J\ell$ and $4J\leq \widetilde{K}\leq 2\ell$. For the special case $\widetilde{K}=2$, which can only occur of $J=1/2$, Equation \eqref{eq:setB} still follows since $|\mathbb{B}_{{\Lambda_\ell,\widetilde{K}}}^j|\leq\ell$. This shows the lemma.

\end{proof}
\subsection{Proof of the main result}
We are now prepared to show our main result:
\begin{thm} \label{thm:main} For any $K\in\{4J^2,4J^2+1,\dots\}$ and any $\delta>0$, we get the following estimate
	\begin{equation} \label{eq:entbound}
	\limsup_{\ell\rightarrow\infty}\:\limsup_{L\rightarrow\infty}\:\frac{\sup_W\:\sup_{\psi}\mathcal{E}(\rho_\psi)}{\log\ell}\leq {2\lfloor K/J\rfloor-2}\:,
	\end{equation}
	where the suprema are taken over all non-negative functions $W:\SSS_L\rightarrow\R_0^+$ and over all normalized elements $\psi\in\mbox{\emph{ran}}(Q_{K,\delta}(L,W))$, respectively.
\end{thm}
\begin{proof} Let $W:\SSS_L\rightarrow\R^+_0$ be an arbitrary non-negative background potential and let $\psi$ be an arbitrary normalized element of $\mbox{ran}(Q_{K,\delta}(L,W))$. Combining Corollary \ref{cor:traceestimate2} and Lemma \ref{lemma:series} (with the choice $\gamma=2\alpha\mu_K$) and using that for any $\alpha\in(0,1)$, the $\alpha$-R\'enyi entropy is an upper bound for the von Neumann entropy, we obtain
	\begin{align}
	\mathcal{E}(\rho_\psi)\leq\frac{1}{1-\alpha}\log\Tr[(\rho_1)^\alpha]\leq \frac{1}{1-\alpha}\log\left(2+2\ell^{4J(1-\alpha)}+(4Je)^{\widetilde{K}-1}C_K'^{2\alpha}\mathcal{L}_{2\alpha\mu_K}^{2J(\widetilde{K}-1)}\ell^{2\widetilde{K}-2}\right)\:,
	\end{align}
which does not depend on $L$. We therefore get
\begin{equation}
\limsup_{\ell\rightarrow\infty}\:\limsup_{L\rightarrow\infty}\:\frac{\sup_W\:\sup_{\psi}\mathcal{E}(\rho_\psi)}{\log\ell}\leq \frac{2\widetilde{K}-2}{1-\alpha}=\frac{2\lfloor {K}/J\rfloor-2}{1-\alpha}\:.
\end{equation}
Since this is true for all $\alpha\in(0,1)$, one can take $\alpha\rightarrow 0$, which yields the desired result.
\end{proof}
\begin{remark}
Note that for $J=1/2$, the constant in \eqref{eq:entbound}, is given by $(4K-2)$. In \cite{ARFS}, where only the case $J=1/2$ was treated, the better bound $(2K-1)$ was established. The main reason for this discrepancy is that in this special case, one can actually show that any configuration ${\bf m}$ with $V({\bf m})\leq K$ can actually be composed out of no more than $K$ building blocks rather than out of no more than ${\widetilde{K}}-1=2K-1$ building blocks.
\end{remark}
\appendix
\section{Proof of the distance formula} \label{app:a}
\begin{lemma}\label{lemma:graphdist} The graph distance $d^N(X,Y)$ from $X$ to $Y$ is given by \eqref{eq:distformula}.
\end{lemma}
\begin{proof} Let $\delta^N(X,Y):=\sum_{i=1}^N|x_i-y_i|$. To prove that $\delta^N(X,Y)$ is a lower bound of the graph distance $d^N(X,Y)$, we first show that ${\bf m}_X\sim{\bf m}_Y$ if and only if $\delta^N(X,Y)=1\;\;\forall\;X,Y\in\SSS_L^N$.
	
	For $X,Y\in\SSS^N_L $, where $X=(x_1,x_2,\dots,x_N)$, $Y=(y_1,y_2,\dots,y_N)$, suppose $\delta^N(X,Y)=1$. By \eqref{eq:distformula}, there exists $k_0\in\{1,2,\dots,N\}$ such that $|x_{k_0}-y_{k_0}|=1$ and $x_k=y_k\;\;\forall\;k\in\{1,2,\dots,N\}\setminus \{k_0\}$ . Define $j_0:=x_{k_0}$. Therefore $y_{k_0}=j_0\pm 1$. Now, let ${\bf m}_X$ and ${\bf m}_Y $ be the occupation number functions corresponding to $X$ and $Y$ respectively. Without loss of generality, let $y_{k_0}= j_0+1$ and suppose that ${\bf m}_X(j_0) = r$, \;\;$r\in \{1,\ldots,2J\}$. Then ${\bf m}_Y(j_0)=r-1$, since $y_{k_0}\neq j_0 $. Hence, ${\bf m}_X(j_0)-{\bf m}_Y(j_0)= r-(r-1)=1$. Also, suppose ${\bf m}_Y(j_0+1)=s,\;\;s\in\set{1,\ldots, 2J}$, then ${\bf m}_X(j_0+1)= s-1$, since $y_{k_0}= j_0+1$. Therefore, ${\bf m}_X(j_0+1)-{\bf m}_Y(j_0+1)=s-1-s=-1$ and ${\bf m}_X(j)={\bf m}_Y(j), \;\;\forall\; j\in \{1,\ldots L\}\setminus \{j_0, j_0+1\}$. Hence, ${\bf m}_X\sim{\bf m}_Y$.
	
	Conversely, suppose ${\bf m}_X\sim{\bf m}_Y$. By \eqref{eq:adjacent}, there exists a unique $j_0\in\{1,\ldots,L-1\}$ such that ${\bf m}_X(j_0)-{\bf m}_Y(j_0)= \pm 1,\;\;\;{\bf m}_X(j_0+1)-{\bf m}_Y(j_0+1)=\mp 1$ and ${\bf m}_X(j)={\bf m}_Y(j)\;\;\forall\;j\in \{1,\ldots, L\}\setminus\{j_0, j_0+1\}$. Without loss of generality, suppose ${\bf m}_X(j_0)-{\bf m}_Y(j_0)=1$ and ${\bf m}_X(j_0+1)-{\bf m}_Y(j_0+1)=-1$. Using the same definitions as before that $j_0:=x_{k_0}$ and ${\bf m}_X(j_0)=r$, let ${\bf m}_X(j_0+1)=s'$, i.e.\ $s'=s-1$ where ${\bf m}_Y(j_0+1)=s$. Let $i_1=\min\set{k : x_k=j_0}$, then $x_k =j_0$ for $k= i_1,\ldots, i_1+(r-1)$ and $x_k= j_0+1$ for $k=i_1+r,\ldots, i_1+r+(s'-1)$. Therefore,
	\begin{equation}\begin{split}
			\delta^N(X,Y)&=\sum_{j=1}^N|x_j-y_j|=0+\cdots+0+|x_{i_1+(r-1)}-y_{i_1+(r-1)}|+0+\cdots+0\\
			&=|j_0-(j_0+1)|=1
	\end{split}\end{equation}
	$\Longrightarrow\;\;$ $\delta^N(X,Y)=1$.\\
	It then follows from the triangle inequality that $\delta^N(X,Y)$ is a lower bound of the graph distance.
	
	To prove equality, it suffices to show that there exists a path from $X$ to $Y$ such that the length of the path is $\delta^N(X,Y)$.
	Let $X,Y\in \SSS_L^N$ and $i_0=\min\{k\in\set{1,\ldots N}: x_k\neq y_k\}$. Without loss of generality, suppose $x_{i_0}<y_{i_0}$. Consider the path
	\begin{equation}\begin{split} Y&=(\ldots,y_{i_0-1},y_{i_0},y_{i_0+1},\ldots)\longrightarrow Y_1=(\ldots, y_{i_0-1},y_{i_0}-1, y_{i_0+1},\ldots)\\
			&\longrightarrow	Y_2=(\ldots, y_{i_0-1},y_{i_0}-2, y_{i_0+1},\ldots)\cdots \longrightarrow Y_{y_{i_0}-x_{i_0}}=(\ldots, y_{i_0-1},x_{i_0}, y_{i_0+1},\ldots)\:.
	\end{split}\end{equation}
	
	The case $x_{i_0}>y_{i_0}$ is similar by switching the roles of $X$ and $Y $. Notice that ${\bf m}_Y\sim{\bf m}_{Y_1}\sim {\bf m}_{Y_2}\sim \cdots\sim {\bf m}_{Y_{y_{i_0}-x_{i_0}}}$ and the length of this path is $|x_{i_0}-y_{i_0}|$.\\
	Define $i_1:=\min\{k\in\set{i_0+1,\ldots,N}: x_k\neq y_k$. Repeating the above process for $x_{i_1}<y_{i_1}$ (with a similar case for $x_{i_1}>y_{i_1}$), we get another path of length $|x_{i_1}-y_{i_1}|$. Since $N<\infty$, the process ends at some $i_s\leq n$ such that $i_s:=\min\{k\in\set{i_0+s,\ldots,N}: x_k\neq y_k\}$ and repeating the process for $x_{i_s}$ yields a path of length $|x_{i_s}-y_{i_s}|$. Therefore, we have a path from $X$ to $Y$ of length 
	\[|x_{i_0}-y_{i_0}|+ |x_{i_1}-y_{i_1}|+\cdots+|x_{i_s}-y_{i_s}|\]
	\begin{equation}\begin{split}
			=\sum_{j:1\leq j<i_0}|x_j-y_j|+|x_{i_0}-y_{i_0}|&+ \sum_{j:i_0< j<i_1}|x_j-y_j|+|x_{i_1}-y_{i_1}|+\cdots\\
			&+ \sum_{j:i_{s-1}< j<i_s}|x_j-y_j|+|x_{i_s}-y_{i_s}|+\sum_{j:i_s< j\leq N}|x_j-y_j|\\
			=\sum^N_{j=1}|x_j-y_j|=\delta^N(X,Y)\:,&
	\end{split}\end{equation}
	which finishes the proof.
\end{proof}
\section{Auxiliary results concerning the interaction potential} \label{app:b}
\subsection{Proof of Proposition \ref{prop:minimizer}}  \label{app:b1}
\begin{proof} Firstly, note that for any ${\bf m}\in{\bf M}_L^N$ one gets
	\begin{align} \label{eq:potentialsimplified}
		V({\bf m})&=\sum_{j=1}^{L-1}\left[J({\bf m}(j)+{\bf m}(j+1))-{\bf m}(j){\bf m}(j+1)\right]+J({\bf m}(1)+{\bf m}(L))\\
		&=2J\sum_{j=1}^L{\bf m}(j)-\sum_{j=1}^{L-1}{\bf m}(j){\bf m}(j+1)=2JN-\sum_{j=1}^{L-1}{\bf m}(j){\bf m}(j+1)\:,
	\end{align}
	which shows that finding minimizers of the potential is equivalent to finding maximizers of $Q_N({\bf m}):=\sum_{j=1}^{L-1}{\bf m}(j){\bf m}(j+1)$. From the explicit form of $Q_N$, it obvious that if $\supp({\bf m})$ is not a discrete interval, i.e. if there exists a $j_0\notin\supp({\bf m})$ with $\min\supp({\bf m})<j_0<\max\supp({\bf m})$, then ${\bf m}$ cannot be a maximizer of $Q_N$. Hence, from now on, we will only consider ${\bf m}\in{\bf M}_L^N$ such that $\supp({\bf m})$ is a discrete interval and w.l.o.g. let us assume $\supp({\bf m})=\{1,2,\dots,k\}$ for some $k\in\N$. For a proof by contradiction, let us now assume that there is an ${\bf m}\in{\bf M}_L^N$ which maximizes $Q_N$ but is not of the form \eqref{eq:form1} or \eqref{eq:form2}. Note that since $N\geq 4J$ but ${\bf m}(j)\leq 2J$ for every $j\in\{1,2,\dots,k\}$, this implies that $k\geq3$. Note that for $k=3$, we get
\begin{equation}
Q_N({\bf m})={\bf m}(1){\bf m}(2)+{\bf m}(2){\bf m}(3)={\bf m}(2)(N-{\bf m}(2))\:,
\end{equation}
where we have used that ${\bf m}(3)=N-{\bf m}(1)-{\bf m}(2)$ for the last equation. Now, observe that since $N\geq 4J$, the maximal possible value for $Q_N$ is only attained if ${\bf m}(2)=2J$, which shows that for $k=3$, any maximizer of $Q_N$ has to be of the form \eqref{eq:form1} or \eqref{eq:form2}. Hence, from now on, we will only consider the case $k\geq 4$.
Let us now distinguish a few different cases and show that for each of these cases, we can construct another $\widehat{\bf m}\in{\bf M}_L^N$ such that $Q_N(\widehat{\bf m})>Q_N({\bf m})$ -- a contradiction to ${\bf m}$ maximizing $Q_N$:
\begin{itemize}
\item[First case:] ${\bf m}(1)>{\bf m}(2)$. In this case, we set
$\widehat{\bf m}(1):={\bf m}(2), \widehat{\bf m}(2):={\bf m}(1)$ and for every other $j$, we define $\widehat{\bf m}(j):={\bf m}(j)$. This yields
\begin{equation}
Q_N(\widehat{\bf m})-Q_N({\bf m})=({\bf m}(1)-{\bf m}(2))\cdot{\bf m}(3)>0\:,
\end{equation}
which shows that ${\bf m}$ cannot maximize $Q_N$ in this case.
\item[Second case:] ${\bf m}(1)\leq {\bf m}(2)$, and it is not true that ${\bf m}(2)={\bf m}(3)= 2J$. We split this second case into two subcases:
\begin{itemize}
\item[First subcase:] ${\bf m}(2)<{\bf m}(3)$. In this case, we define the configuration $\widehat{\bf m}$ by setting $\widehat{\bf m}(1):={\bf m}(1)-1, \widehat{\bf m}(2):={\bf m}(2)+1$ and for every other $j$, we define $\widehat{\bf m}(j):={\bf m}(j)$. This yields
\begin{equation}
Q_N(\widehat{\bf m})-Q_N({\bf m})=({\bf m}(1)-1)+{\bf m}(3)-{\bf m}(2)\geq {\bf m}(3)-{\bf m}(2)>0\:,
\end{equation}
which shows that ${\bf m}$ cannot maximize $Q_N$ in this case.
\item[Second subcase:] ${\bf m}(2)\geq {\bf m}(3)$. Note that since we are in the Second case, it is not possible that ${\bf m}(3)=2J$. We then define the configuration $\widehat{\bf m}$ by setting $\widehat{\bf m}(1):={\bf m}(1)-1, \widehat{\bf m}(3):={\bf m}(3)+1$ and for every other $j$, we define $\widehat{m}(j):={\bf m}(j)$. We then get
\begin{equation}
Q_N(\widehat{\bf m})-Q_N({\bf m})={\bf m}(4)>0\:,
\end{equation}
which shows that ${\bf m}$ is not a maximizer of $Q_N$.
\end{itemize}
\item[Third case:] ${\bf m}(2)={\bf m}(3)=2J$. Note that in this case, we must have $k\geq 5$, since otherwise ${\bf m}$ would be of the form \eqref{eq:form1} or \eqref{eq:form2}. Let $k_0\in\{4,\dots,k-1\}$ be the smallest number for which ${\bf m}(k_0)\neq 2J$. (Again, note that such a $k_0$ must exist, since otherwise ${\bf m}$ would be of the form \eqref{eq:form1} or \eqref{eq:form2}.) In particular, this means that ${\bf m}(k_0-1)=2J$ and ${\bf{m}}(k_0+1)>0$. Now, define $\widehat{\bf m}(1):={\bf m}(1)-1, \widehat{\bf m}(k_0):={\bf m}(k_0)+1$ and for every other $j$, we set $\widehat{\bf m}(j):={\bf m}(j)$. We then find
\begin{equation}
Q_N(\widehat{\bf m})-Q_N({\bf m})={\bf m}(k_0+1)>0\:,
\end{equation}
which shows that ${\bf m}$ cannot be a maximizer of $Q_N$.
\end{itemize}
Since these three cases exhaust all possibilities, this finishes the proof.
\end{proof}
\subsection{Proof of Lemma \ref{lemma:localminimizer}} \label{app:b2}
\begin{proof}
	Let $Z=(z_1,z_2,\dots,z_{j+k})\in\SSS_{L,K}^{j+k}$ be such that $d^{j+k}(X\cup Y,\SSS^{j+k}_{L,K})=d^{j+k}(X\cup Y,Z)$. Now, introduce the configurations ${Z}'=(z_1,z_2,\dots,z_j)\in\SSS_L^{j}$ and ${Z}''=(z_{j+1},\dots,z_{j+k})\in\SSS_L^{k}$. Moreover, let ${\bf m}_{Z'}$ and ${\bf m}_{Z''}$ denote the corresponding occupation number functions. Note that ${\bf m}_{Z'}(i)=0$ for $i>z_j$ as well as ${\bf m}_{Z''}(i)=0$ for $i<z_j$. Moreover, note that ${\bf m}_{Z'}(z_j)+{\bf m}_{Z''}(z_j)\leq 2J$. It follows that $V({\bf m}_{Z'})\leq V({\bf m}_{Z'}+{\bf m}_{Z''})$, i.e. $V({Z}')\leq V(Z)$ and thus $Z'\in\SSS_{L,K}^j$. To see this, observe that
	\begin{align}
		V({\bf m}_{Z'}+{\bf m}_{Z''})&=V({\bf m}_{Z'})+(2J-{\bf m}_{Z'}(z_j-1)){\bf m}_{Z''}(z_j)+(2J-({\bf m}_{Z'}(z_j)+{\bf m}_{Z''}(z_j)){\bf m}_{Z''}(z_j+1)\\&+\sum_{i=z_j+1}^{L-1}(2J-{\bf m}_{Z''}(i)){\bf m}_{Z''}(i+1)\geq V({\bf m}_{Z'})\:,
		\label{eq:cutoffpotential}
	\end{align}
	where we used that since ${\bf m}_{Z'},{\bf m}_{Z''}\in{\bf M}_L$, one has ${\bf m}_{Z'}(i),{\bf m}_{Z''}(i)\leq 2J$ for any $i\in\{1,2,\dots,L\}$ and moreover that ${\bf m}_{Z'}(z_j)+{\bf m}_{Z''}(z_j)\leq 2J$.
	Now, if in addition $z_j\leq \ell$, this would imply $Z'\in\SSS_{\Lambda_\ell,K}^j$ and since -- trivially -- $d^j(X,\SSS_{\Lambda_\ell,K}^j)\leq d^j(X,Z')\leq d^{j+k}(X\cup Y,Z'\cup Z'')=d^{j+k}(X\cup Y,\SSS_{L,K}^{j+k})$ this shows \eqref{eq:distest1} for the case $z_j\leq \ell$. If $z_j>\ell$, we need to modify our argument. If $z_1>\ell$, we construct the occupation number function ${\bf g}\in{\bf M}^j_{\Lambda_\ell}$ as follows:
	\begin{equation}
	{\bf g}(i):= \begin{cases}
	2J&\mbox{if}\qquad\ell-\lfloor{\frac{j}{2J}}\rfloor+1\leq i\leq\ell\\
	j(\bmod\;2J)&\mbox{if}\qquad \ell-\lfloor{\frac{j}{2J}}\rfloor
	\end{cases}\:.
	\end{equation}
	Let $X_{\bf g}\in \SSS^j_{\Lambda_\ell}$ denote the multiset associated with $\bf g$. Note that $X_{\bf g}$ is of the form \eqref{eq:form1} or \eqref{eq:form2} and thus, by Proposition \ref{prop:minimizer} we have $X_{\bf g}\in \SSS^j_{\Lambda_\ell, K}$. Therefore, $d^j(X, \SSS^j_{\Lambda_\ell,K})\leq d^j(X,X_{\bf g})\leq d^j(X,Z')\leq d^{j+k}(X\cup Y,\SSS^{j+k}_{L,K})$ (where the second inequality is due to the fact that $x_i\leq g_i\leq \ell < z_i, \;\:i=1,\ldots, j$). So, from now on, we may assume that $z_1\leq \ell$ and thus let $j_0:=\max\{i\in\{1,2,\dots,j\}: z_i\leq \ell\}$ be the index of the last particle of $Z'$ that still lies in $\Lambda_\ell$. We now split $Z'=Z'_{\Lambda_\ell}\cup Z'_{\Lambda_\ell^c}$, where $Z'_{\Lambda_\ell}=(z_1,\dots,z_{j_0})$ and $Z'_{\Lambda_\ell^c}=(z_{j_0+1},\dots,z_j)$. Let ${\bf m}_{Z_{\Lambda_\ell}'}\in{\bf M}_{\Lambda_\ell}^{j_0}$ be the occupation number function associated to $Z_{\Lambda_\ell}'$. Lastly, let $j_c:=j-j_0$ be the number of particles in $Z'_{\Lambda_\ell^c}$. We then distinguish the following two cases:
\begin{itemize}
\item[Case 1:] $j_c+{\bf m}_{Z_{\Lambda_\ell}'}(\ell)+{\bf m}_{Z_{\Lambda_\ell}'}(\ell-1)\geq 4J$.\\
Define $r_0 := \max\{k\in\{1,\ldots, \ell-2\}:\sum_{s=k}^\ell(2J-{\bf m}_{Z_{\Lambda_\ell}'}(s))\geq j_c\}$ and let $\mu=j_c -\sum^\ell_{s=r_0+1}(2J-{\bf m}_{Z_{\Lambda_\ell}'}(s))$. We construct the occupation number function ${\bf c}\in{\bf M}_{\Lambda_\ell}^j$ as follows:
\begin{equation}
	    {\bf c}(i):=\begin{cases} {\bf m}_{Z_{\Lambda_\ell}'}(i)&\mbox{if}\qquad 1\leq i\leq r_0-1\\
	    {\bf m}_{Z_{\Lambda_\ell}'}(i)+\mu&\mbox{if}\qquad i=r_0\\
	    2J&\mbox{if}\qquad r_0+1\leq i \leq \ell
	    \end{cases}
	\end{equation}
In other words, the particles in $Z'_{\Lambda_\ell^c}$ are used to ``fill up" each site of configuration ${\bf m}_{Z_{\Lambda_\ell}'}$ to its maximal occupation number $2J$ -- starting from $\ell, \ell-1,\mbox{etc.}$ until the $j_c$ particles in $Z'_{\Lambda_\ell^c}$ have been exhausted (cf.\ Figure \ref{moveparticle}). Note that the choice ${\bf c}(\ell)={\bf c}(\ell-1)=2J$ is following the same principle and is due to the assumption $j_c+{\bf m}_{Z_{\Lambda_\ell}'}(\ell)+{\bf m}_{Z_{\Lambda_\ell}'}(\ell-1)\geq 4J$. 
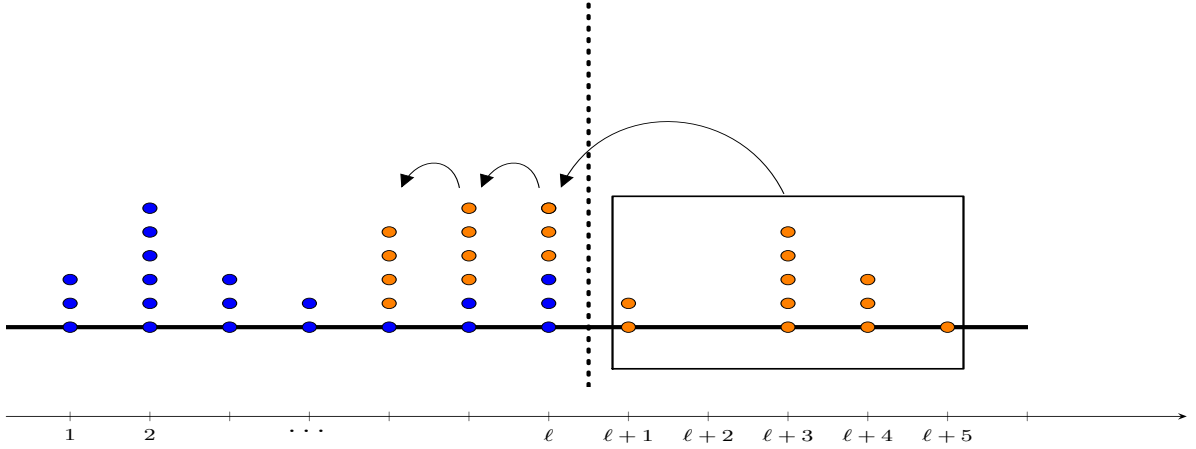
\begin{figure}[ht]
	\centering
	\resizebox{16cm}{6cm}
	{

\begin{tikzpicture}[line cap=round,line join=round,>=triangle 45,x=1cm,y=1cm]
\begin{axis}[
x=1cm,y=1cm,
axis x line=bottom,
axis y line=none,
xmin=0.2,
xmax=15,
ymin=-0.5,
ymax=6.5,
xtick={1,...,13},
xticklabels={\tiny{1},\tiny{2},,\ldots,,, \tiny{$\ell$},\tiny{$\ell+1$},\tiny{$\ell+2$},\tiny{$\ell+3$},\tiny{$\ell+4$},\tiny{$\ell+5$}}
]
\clip(0,0) rectangle (13,8.7);
\draw [line width=2pt] (0,1)-- (15,1);
\draw [loosely dotted,line width=1.5pt] (7.5,0)-- (7.5,6.5);
\begin{scriptsize}
\draw [fill=qqqqff] (1,1) circle (2.5pt);
\draw [fill=qqqqff] (1,1.4) circle (2.5pt);
\draw [fill=qqqqff] (1,1.8) circle (2.5pt);
\draw [fill=qqqqff] (2,1) circle (2.5pt);
\draw [fill=qqqqff] (2,1.4) circle (2.5pt);
\draw [fill=qqqqff] (2,1.8) circle (2.5pt);
\draw [fill=qqqqff] (2,2.2) circle (2.5pt);
\draw [fill=qqqqff] (2,2.6) circle (2.5pt);
\draw [fill=qqqqff] (2,3) circle (2.5pt);
\draw [fill=qqqqff] (3,1) circle (2.5pt);
\draw [fill=qqqqff] (3,1.4) circle (2.5pt);
\draw [fill=qqqqff] (3,1.8) circle (2.5pt);
\draw [fill=qqqqff] (4,1) circle (2.5pt);
\draw [fill=qqqqff] (4,1.4) circle (2.5pt);
\draw [fill=qqqqff] (5,1) circle (2.5pt);
\draw [fill=orange] (5,1.4) circle (2.5pt);
\draw [fill=orange] (5,1.8) circle (2.5pt);
\draw [fill=orange] (5,2.2) circle (2.5pt);
\draw [fill=orange] (5,2.6) circle (2.5pt);
\draw [fill=qqqqff] (6,1) circle (2.5pt);
\draw [fill=qqqqff] (6,1.4) circle (2.5pt);
\draw [fill=orange] (6,1.8) circle (2.5pt);
\draw [fill=orange] (6,2.2) circle (2.5pt);
\draw [fill=orange] (6,2.6) circle (2.5pt);
\draw [fill=orange] (6,3.0) circle (2.5pt);
\draw [fill=orange] (7,3.0) circle (2.5pt);
\draw [fill=qqqqff] (7,1) circle (2.5pt);
\draw [fill=qqqqff] (7,1.4) circle (2.5pt);
\draw [fill=qqqqff] (7,1.8) circle (2.5pt);
\draw [fill=orange] (7,2.2) circle (2.5pt);
\draw [fill=orange] (7,2.6) circle (2.5pt);
\draw [fill=orange] (7,3.0) circle (2.5pt);
\draw [fill=orange] (8,1) circle (2.5pt);
\draw [fill=orange] (8,1.4) circle (2.5pt);
\draw [fill=orange] (10,1) circle (2.5pt);
\draw [fill=orange] (10,1.4) circle (2.5pt);
\draw [fill=orange] (10,1.8) circle (2.5pt);
\draw [fill=orange] (10,2.2) circle (2.5pt);
\draw [fill=orange] (10,2.6) circle (2.5pt);
\draw [fill=orange] (11,1) circle (2.5pt);
\draw [fill=orange] (11,1.4) circle (2.5pt);
\draw [fill=orange] (11,1.8) circle (2.5pt);
\draw [fill=orange] (12,1) circle (2.5pt);
\node (f) at (5.1, 3.2) {};
\node (e) at (5.9,3.2) {};
\draw[->] (e)  to [out=98,in=70, looseness=2.0] (f);
\node (d) at (6.1, 3.2) {};
\node (c) at (6.9,3.2) {};
\draw[->] (c)  to [out=98,in=70, looseness=2.0] (d);
\node (b) at (7.1, 3.2) {};
\node (a) at (10, 3.1) {};
\draw[->] (a)  to [out=110,in=70, looseness=1.5] (b);
\draw[black,thick] (7.8,0.3) -- (12.2,0.3) -- (12.2, 3.2) -- (7.8,3.2) -- (7.8,0.3);
\end{scriptsize}
\end{axis}
\end{tikzpicture}

}
\caption{Picture showing the construction of configuration ${\bf c}$, by moving the $j_c=11$ particles in $\Lambda_\ell^c$ such that they fill up each site to its maximum occupation number $2J=6$ starting at $\ell$ to the left (here: from $\ell$ to $\ell-2$).}
\label{moveparticle}
\end{figure}

\noindent
We now make the following observations: Firstly, let $X_{\bf c}\in\SSS_{\Lambda_\ell}^j$ denote the multiset corresponding to ${\bf c}$. Then, by construction, it is clear that $d^j(X,X_{\bf c})\leq d^j(X,Z')$. Secondly, let $s_0=\min\{s\in\{1,2,\dots,\ell-2\}:\forall t\in\{s+1,\dots,\ell\}: {\bf c}(t)=2J\}$. If $s_0=1$, then $X_{\bf c}$ is a minimizer configuration with ${\bf c}(i)=2J\quad \forall\;i= 2,\ldots,\ell$ and hence, $X_{\bf c}\in \SSS^j_{\Lambda_\ell,K}$. Since $x_i\leq c_i\leq z_i,\;\; i=1,\ldots,j$, we get
\begin{equation}
d^j(X,X_{\bf c})= \sum^j_{i=1}(c_i-x_i)\leq \sum^j_{i=1}(z_i-x_i)\leq d^j(X,Z') 
\end{equation}
and therefore,
$d^j(X,\SSS^j_{\Lambda_\ell,K})\leq d^j(X,X_{\bf c})\leq d^j(X,Z')\leq d^{j+k}(X\cup Y,\SSS^{j+k}_{L,K})$ for the case $s_0=1$.   

Now for $s_0\neq 1$, define ${\bf c}'':={\bf c}\cdot1_{\{s_0,\dots,\ell\}}$, ${\bf c}':={\bf c}-{\bf c}''$, ${\bf d}'':={ {\bf m}_{Z'}}\cdot1_{\{s_0,\dots,L\}}$ and ${\bf d}':={\bf m}_{Z'}-{\bf d}''$.
		Then, observe that ${\bf c}(s_0)\geq {\bf m}_{\Z'_{\Lambda_\ell}}(s_0)$ and moreover that ${\bf c}(s)= {\bf m}_{\Z'_{\Lambda_\ell}}(s)$ for every $s<s_0$. Consequently, we get ${\bf c}'={\bf d}'$. Observe that ${\bf c}''$ is of the form \eqref{eq:form1} or \eqref{eq:form2} and thus $V({\bf c}'')=4J^2$. Moreover, by Proposition \ref{prop:minimizer}, we have $V({\bf c''})\leq V({\bf d''})$. 
		Consequently, we get
		\begin{equation}
		V({\bf c})=V({\bf c'})-{\bf c}(s_0-1){\bf c}(s_0)+V({\bf c}'')\leq V({\bf d}')-{\bf m}_{Z'_{\Lambda_\ell}}(s_0-1){\bf m}_{Z'_{\Lambda_\ell}}(s_0)+V({\bf d}'')=V({\bf m}_{Z'})\leq K\:.
		\end{equation}
		This shows that $X_{\bf c}\in\SSS_{\Lambda_\ell,K}^j$. Since $d^j(X,\SSS_{\Lambda_\ell,K}^j)\leq d^j(X,X_{\bf c})\leq d^j(X,Z')\leq d^{j+k}(X\cup Y,\SSS_{L,K}^{j+k})$, this shows the assertion for Case 1.\\
		
		\item[Case 2:] $j_c+{\bf m}_{Z_{\Lambda_\ell}'}(\ell)+{\bf m}_{Z_{\Lambda_\ell}'}(\ell-1)< 4J$. 
		Firstly, let us define the following two quantities:
		\begin{align}
			\xi'&:={\bf m}_{Z'_{\Lambda_\ell}}(\ell-1)+\sum_{k\in\N: \ell+2k-1\leq L}{\bf m}_{Z'_{\Lambda^c_\ell}}(\ell+2k-1)\\
			\xi''&:={\bf m}_{Z'_{\Lambda_\ell}}(\ell)\qquad+\sum_{k\in\N: \ell+2k\leq L}{\bf m}_{Z'_{\Lambda^c_\ell}}(\ell+2k)\quad\mbox{and observe that}\\
			&\sum_{i=\ell-1}^{L-1}{\bf m}_{Z'}(i){\bf m}_{Z'}(i+1)\leq \xi'\xi''\:, \label{eq:xixi}
		\end{align}
		since the left hand side of \eqref{eq:xixi} is a sum over a subset of the (non--negative) cross--terms one obtains when expanding the product $\xi'\xi''$.
		Moreover, note that the assumption for being in Case 2 is equivalent to $\xi'+\xi''<4J$.
		Next, we define the occupation number function ${\bf e}\in\SSS_{\Lambda_\ell}^j$ as follows:
		\begin{equation}
		{\bf e}(i):=\begin{cases}\min\{2J,\xi''\}+\max\{\xi'-2J,0\}\qquad&\mbox{if}\qquad i=\ell\\
		\min\{2J,\xi'\}+\max\{\xi''-2J,0\}\qquad&\mbox{if}\qquad i=\ell-1\\
		{\bf m}_{Z'_{\Lambda_\ell}}(i)\qquad&\mbox{if}\qquad i\in\{1,2,\dots,\ell-2\}\:.
		\end{cases}
		\end{equation}
		Let us now show that $V({\bf e})\leq V({\bf m}_{Z'})$. To this end, observe that
		\begin{align} 
			V({\bf m}_{Z'})- V({\bf e})=&{\bf m}_{Z'}(\ell-2)({\bf e}(\ell-1)-{\bf m}_{Z'}(\ell-1))+{\bf e}(\ell-1){\bf e}(\ell)-\sum_{i=\ell-1}^{L-1} {\bf m}_{Z'}(i){\bf m}_{Z'}(i+1)\\ \geq& {\bf e}(\ell-1){\bf e}(\ell)-\xi'\xi''\:,
			\label{eq:potentialest}	\end{align}
		where we used ${\bf e}(\ell-1)\geq {\bf m}_{Z_{\Lambda_\ell}}(\ell-1)$ and \eqref{eq:xixi}.
		Since $\xi'+\xi''<4J$, there are only three possible cases:
		\begin{itemize}
			\item[(i)] $\xi''>2J$ and $\xi'\leq 2J$
			\item[(ii)]$\xi''\leq 2J$ and $\xi'>2J$
			\item[(iii)] $\xi'\leq 2J$ and $\xi''\leq 2J$.
		\end{itemize}
		We will only discuss Cases (i) and (iii); Case (ii) follows from an argument similar to Case (i). If we are in Case (i), this means ${\bf e}(\ell)=2J$ and ${\bf e}(\ell-1)= \xi'+\xi''-2J$. Hence, by \eqref{eq:potentialest}, we get
		\begin{equation}
		V({\bf m}_{Z'})-V({\bf e})\geq 2J(\xi'+\xi''-2J)-\xi'\xi''=(2J-\xi')(\xi''-2J)\geq 0\:.
		\end{equation}
		For Case (iii), we have ${\bf e}(\ell)=\xi''$ and ${\bf e}(\ell-1)=\xi'$ and thus get -- again by \eqref{eq:potentialest} -- the estimate
		\begin{equation}
		V({\bf m}_{Z'})-V({\bf e})\geq \xi'\xi''-\xi'\xi''=0\:.
		\end{equation}
		Hence, we have shown $V({\bf e})\leq V({\bf m}_{Z'}) $. 
		Now, let $X_{\bf e}\in\SSS_{\Lambda_\ell}^j$ denote the multiset associated to ${\bf e}$. If we can show that $d^j(X,X_{\bf e})\leq d^j(X,Z')$, this will prove \eqref{eq:distest1} for the Case 2. Now, to see this, we firstly define the configuration ${\bf f}\in{\bf M}_{\Lambda_\ell}^j$ as follows
		\begin{equation}
		{\bf f}(i):=\begin{cases} \min\{2J,{\bf m}_{Z'}(\ell)+j_c\}\quad&\mbox{if}\quad i=\ell\\
		{\bf m}_{Z'}(\ell-1)+\max\{{\bf m}_{Z'}(\ell)+j_c-2J,0\}\quad&\mbox{if}\quad i=\ell-1\\
		{\bf m}_{Z'}(i)\quad&\mbox{if}\quad i\in\{1,2,\dots,\ell-2\}\:,
		\end{cases}
		\end{equation}
		which means that configuration ${\bf f}$ is obtained by adding the particles in $Z'_{\Lambda_c}$ to the configuration ${\bf m}_{Z'}$ -- starting at site $\ell$ and any possibly remaining particles to site $\ell-1$. Let $X_{\bf f}=(f_1,f_2,\dots,f_j)\in\SSS_{\Lambda_\ell}^j$ be the multiset associated to ${\bf f}$ and let $p=\min\{i:z_i>\ell\}$ and $\xi=\max\set{i: z_i\leq \ell-2}$ (note that $\xi$ exists since if $z_1>\ell-2$ and $j_c< 4J-{\bf m}_{Z'}(\ell-1)-{\bf m}_{Z'}(\ell)$, then we would have $j=j_0+j_c< 4J$, a contradiction). Again, we distinguish two cases:
\begin{itemize}
	\item [(a)] If $j_c\leq 2J-{\bf m}_{Z'}(\ell)$, then $z_i=f_i$ for $1\leq i\leq p-1$ and $f_i=\ell $ for $p\leq  i\leq j$. Then,
	\begin{equation}
	\begin{split}
	d^j(X,Z')&=\sum_{i=1}^{p-1}|x_i-z_i|+\sum_{i=p}^j (z_i-x_i)\\
	&=\sum_{i=1}^{p-1}|x_i-z_i|+\sum_{i=p}^j (z_i-f_i)+\sum_{i=p}^j (f_i-x_i)\\
	&=\sum_{i=1}^{p-1}|x_i-f_i|+\sum_{i=p}^j |z_i-f_i|+\sum_{i=p}^j |f_i-x_i|\\
	&= d^j(X,X_{\bf f})+\sum_{i=j-j_c+1}^j|z_i-f_i|\geq d^j(X,X_{\bf f})+ j_c
	\end{split}
	\end{equation}
	\item[(b)] If $j_c> 2J-{{\bf m}_{Z'}}(\ell)$, let $\eta= j_c-(2J-\alpha_1)$, where $\alpha_1={{\bf m}_{Z'}}(\ell)$ and $\alpha_2={{\bf m}_{Z'}}(\ell-1)$. Then, observe that
\end{itemize}
\begin{equation}
	\left\{\begin{split}
	f_i &=z_i,\;\;\;\;\;1\leq i\leq \xi+\alpha_2\\
	f_i& = \ell-1,\;\;\;\;\;\xi+\alpha_2+1\leq i\leq \xi+\alpha_2+n\\
	f_i&=\ell,\;\;\;\;\; i>\xi+\alpha_2+\eta\\
	z_i&=\ell,\;\;\;\;\; \xi+\alpha_2+1\leq i \leq p-1
	\end{split}\right.	
\end{equation}
Hence,
\begin{equation}\begin{split}
d^j(X,Z')&=\sum_{i=1}^{\xi+\alpha_2}|x_i-z_i|+\sum_{i=\xi+\alpha_2+1}^{p-1} |x_i-z_i|+\sum_{i=p}^j(z_i-x_i)\\\\
&=\sum_{i=1}^{\xi+\alpha_2}|x_i-z_i|+\sum_{i=\xi+\alpha_2+1}^{\xi+\alpha_2+\eta} (z_i-x_i)+\sum_{i=\xi+\alpha_2+\eta+1}^{p-1} (z_i-x_i)+\sum_{i=p}^j(z_i-f_i)+\sum_{i=p}^j(f_i-x_i)\\\\
&=\sum_{i=1}^{\xi+\alpha_2}|x_i-f_i|+\sum_{i=\xi+\alpha_2+1}^{\xi+\alpha_2+\eta}(z_i-f_i)+\sum_{i=\xi+\alpha_2+1}^{\xi+\alpha_2+\eta}(f_i-x_i)+\sum_{i=\xi+\alpha_2+\eta+1}^{p-1}(z_i-f_i)\\
&\;\;\;\;\;\;\;\;\;\;\;\;\;\;\;\;\;+\sum_{i=\xi+\alpha_2+\eta+1}^{p-1} (f_i-x_i)+\sum_{i=p}^j(z_i-f_i)+\sum_{i=p}^j(f_i-x_i)
\end{split}
\end{equation}
\begin{equation*}
\begin{split}
&=\sum_{i=1}^{\xi+\alpha_2}|x_i-f_i|+\sum_{i=\xi+\alpha_2+1}^{\xi+\alpha_2+\eta}|f_i-x_i|+\sum_{i=\xi+\alpha_2+\eta+1}^{p-1}|f_i-x_i|+\sum_{i=p}^{j} |f_i-x_i|\\
&\;\;\;\;\;\;\;\;\;\;\;\;\;\;\;\;\;+\sum_{i=\xi+\alpha_2+1}^{\xi+\alpha_2+\eta} |z_i-f_i|+\sum_{i=\xi+\alpha_2+\eta+1}^{p-1}|z_i-f_i|+\sum_{i=p}^j|z_i-f_i|\\\\
&=d^j(X,X_{\bf f})+\sum_{i=\xi+\alpha_2+1}^{\xi+\alpha_2+\eta} |z_i-f_i|+\sum_{i=\xi+\alpha_2+\eta+1}^{p-1}|z_i-f_i|+\sum_{i=p}^j|z_i-f_i|\\\\
&\geq d^j(X,X_{\bf f}) + \sum_{i=p}^j|z_i-f_i|\\
&\geq d^j(X,X_{\bf f})+ (j-p+1)= d^j(X,X_{\bf f}) +j_c
\end{split}
\end{equation*}
Moreover, since ${\bf e}(i)\neq{\bf f}(i)$ only if $i=\ell-1,\ell$, this implies
\begin{equation}
d^j(X_{\bf e},X_{\bf f})=|{\bf f}(\ell)-{\bf e}(\ell)|\leq j_c\:.
\end{equation}
We therefore get
\begin{equation}
d^j(X,X_{\bf e})\leq d^j(X,X_{\bf f})+d^j(X_{\bf f},X_{\bf e})\leq d^j(X,X_{\bf f})+j_c\leq d^j(X,Z')\;,
\end{equation}
and thus, we have shown \eqref{eq:distest1} for Case 2.
\end{itemize}
This finishes the proof.
\end{proof}

\end{document}